\documentclass[format=acmsmall, review=false]{acmart}

\usepackage{acm-ec-19}
\usepackage{booktabs} 

\usepackage{tikz}

\begin{document}
\title[Dividing a Graphical Cake]{Dividing a Graphical Cake}
\author{Xiaohui Bei}
\affiliation{%
  \institution{Nanyang Technological University}
  \country{Singapore}}
\author{Warut Suksompong}
\affiliation{%
  \institution{University of Oxford}
  \country{UK}}

\begin{abstract}
We consider the classical cake-cutting problem where we wish to fairly divide a heterogeneous resource, often modeled as a cake, among interested agents.
Work on the subject typically assumes that the cake is represented by an interval.
In this paper, we introduce a generalized setting where the cake can be in the form of the set of edges of an undirected graph. 
This allows us to model the division of road or cable networks.
Unlike in the canonical setting, common fairness criteria such as proportionality cannot always be satisfied in our setting if each agent must receive a connected subgraph.
We determine the optimal approximation of proportionality that can be obtained for any number of agents with arbitrary valuations, and exhibit tight guarantees for each graph in the case of two agents.
In addition, when more than one connected piece per agent is allowed, we establish the best egalitarian welfare guarantee for each total number of connected pieces.
We also study a number of variants and extensions, including when approximate equitability is considered, or when the item to be divided is undesirable (also known as chore division).
\end{abstract}

\maketitle


\section{Introduction}

Cake cutting refers to the problem of fairly allocating a divisible resource, often modeled as a cake, among agents with varying preferences. The problem dates back to shortly after the end of World War II \citep{Steinhaus48} and, not surprisingly given its wide range of applications, still enjoys significant attention in mathematics, computer science, economics, and political science to this day \citep{BramsTa96,RobertsonWe98,Procaccia16}.

What does it mean for an allocation to be fair? 
Steinhaus \citet{Steinhaus48} proposed the following definition of fairness: if the cake is divided between $n$ agents, each agent should receive a part that she values at least $1/n$ of the entire cake. 
This definition became known as \emph{proportionality}, and is one of the most fundamental notions in the literature of fair division. 
In his seminal article, Steinhaus showed that a proportional allocation can be found for any number of agents with arbitrary preferences over the cake---Steinhaus' method, which he attributed to Knaster and Banach, was later formulated as a moving-knife procedure by Dubins and Spanier \citet{DubinsSp61}. 
The procedure works by having a referee move a knife over the cake from left to right. 
Whenever the left part has value $1/n$ of the entire cake for one of the agents, the agent takes that part of the cake and leaves; the procedure is then repeated among the remaining agents. 
In addition to ensuring proportionality, the Dubins-Spanier protocol has the important property that it always allocates to every agent a \emph{connected} piece of the cake. 
Without this property, it may well be that an agent is presented with---in the words of Stromquist \cite{Stromquist80}---a ``union of crumbs''.

The proportionality guarantee of the Dubins-Spanier protocol holds for a variety of items that one may wish to divide, since an item of any shape or form can be ``projected'' onto a line, which we can then run the protocol on. 
However, the connectivity property does not necessarily translate from the line back to the original shape. 
A simple illustrating example is when the item has the shape of a ring (e.g., a donut or a ring road): projecting the ring onto a line and applying Dubins-Spanier may result in an agent receiving disconnected pieces of the ring. 
For this example, the difficulty can be circumvented by cutting the ring at an arbitrary point, stretching it into a line, and running the protocol to achieve proportionality. 
Nevertheless, one may already begin to suspect that this type of fix no longer works when the shapes get more complex. 

In this paper, we consider a natural setting where the cake is represented by the set of edges of an arbitrary undirected graph. 
The graph could correspond to a resource in the form of a network, such as a road or power cable network. 
The canonical cake-cutting setting, where the cake is assumed to be the interval $[0,1]$, is a special case of our setting, with the graph consisting of a single edge.
In this generalized graphical setting, we show that proportionality cannot always be attained if connectivity is required.
Therefore, our goal in this work is to provide the best possible approximation of proportionality that can be achieved in each graph. 
As we will see, despite allowing arbitrary graphs and agents' preferences in our model, we can still obtain several strong fairness guarantees.
Furthermore, we study a number of variants and extensions, including when each agent can get more than one connected piece, or when the item to be divided is undesirable (often referred to as a \emph{chore}).

\subsection{Our Results}

We assume throughout the paper that the agents are endowed with valuation functions that are additive and normalized with each agent having value $1$ for the whole cake; these assumptions are standard in the cake-cutting literature \citep{Procaccia16}. 
We also assume that the pieces of cake that different agents receive may intersect in a finite number of points.
Our formal model is described in Section~\ref{sec:prelim}.

In Section~\ref{sec:anyagents}, we provide general guarantees that hold for any number of agents. 
For $n$ agents with arbitrary valuations and any graph, we establish the existence of a connected allocation that gives every agent a utility of at least $\frac{1}{2n-1}$. We also show that this bound is tight even when the agents have identical valuations and the graph is a star with $2n-1$ edges. 
In addition, for every specific star graph, we determine the optimal utility that can be guaranteed for each number of agents.

In Section~\ref{sec:twoagents}, we delve deeper into the case of two agents.
While we know from Section~\ref{sec:anyagents} that both agents can be guaranteed a utility of $1/3$ in general, for certain graphs it is possible to do better. 
Perhaps surprisingly, we show that the optimal guarantee for each graph is always either $1/3$ or $1/2$; the latter case corresponds to a proportional allocation. The classification depends on a graph property that we call \emph{almost bridgeless}---a graph satisfies this property if we can add an edge so that the resulting graph contains no bridges, where a bridge refers to an edge that is not contained in any cycle. 
We show that a guarantee of $1/2$ can be obtained if the graph is almost bridgeless, while $1/3$ is the best possible guarantee otherwise.

Next, we strengthen the result that both agents can be guaranteed a utility of $1/3$ for any graph by establishing the existence of a connected allocation such that the first agent receives value at least $1/2$ and the second agent at least $1/3$. 
More generally, we characterize all pairs $(\alpha,\beta)$ for which there always exists a connected allocation that yields utility at least $\alpha$ and $\beta$ to the first and second agent, respectively.
On the other hand, if we are only interested in giving utility $\alpha$ to one agent and $\beta$ to the other, and are willing to give up control over which agent receives which guarantee, then additional pairs $(\alpha,\beta)$ become achievable---again, we give a complete characterization of all such pairs.
Moreover, we consider allocating more than one connected piece to each agent: we show that for any positive integer $k$, both agents can be guaranteed a utility of $\frac{1}{2}-\frac{1}{2\cdot 3^k}$ if we allow the agents to receive a total of $k+1$ connected pieces, and this is tight.
We also study approximate \emph{equitability} and prove the existence of a connected allocation for which the agents' utilities differ by no more than $1/3$; we again establish the tightness of the bound.

Finally, in Section~\ref{sec:choredivision}, we turn our attention to chore division. 
In the case of two agents there is a simple reduction between the settings of cake and chores, so all of our results in Section~\ref{sec:twoagents} carry over to chore division. 
By contrast, when there are more than two agents, the relationship between the two settings is much less clear. 
We show that there exists a connected allocation that incurs cost at most $\frac{2}{n+1}$ for $n\leq 5$, and that no better bound can be obtained for any $n$.

We remark that all of our positive results are constructive: for each result, we exhibit a moving-knife protocol that achieves the desired guarantee. 
Moreover, one can make these protocols discrete, so that they only use the \emph{cut} and \emph{evaluation} queries allowed by the Robertson-Webb query model in order to access the valuation functions of the agents \citep{RobertsonWe98}.

\subsection{Further Related Work}

While our graphical cake model is new to the best of our knowledge, graphs have recently been studied in the context of allocating \emph{indivisible} items. 
In particular, the items correspond to vertices of an undirected graph, and each agent must be allocated a connected subgraph of the graph. 
A line of work has explored existence and complexity questions for several fairness notions, both in the case of goods \citep{BouveretCeEl17,LoncTr18,BiloCaFl19,BeiIgLu19,IgarashiPe19,Suksompong19} and chores \citep{BouveretCeLe19}.

Connectivity constraints are commonly considered in the cake-cutting literature, where each agent is allocated a single subinterval of the interval cake \citep{Stromquist80,Stromquist08,Su99,BeiChHu12,CechlarovaPi12,AumannDoHa13,CechlarovaDoPi13,AumannDo15,SegalhaleviHaAu16}.
Segal-Halevi et al. \cite{SegalhaleviNiHa17} studied the fair division of land and introduced geometric constraints to the setting by requiring that allocated pieces be of certain shape---such requirements are important since a long but narrow piece of land is likely to be of little practical use. 
Similarly to our setting, a proportional allocation does not always exist in the presence of these constraints, and the authors examined the approximations of proportionality that can be obtained.

\section{Preliminaries}
\label{sec:prelim}

Let $N=\{1,2,\dots,n\}$ be the set of agents, and $G=(V,E)$ be a finite and connected undirected graph representing the cake, with no loops but possibly with multiple edges joining the same pair of vertices.\footnote{A loop can be represented in our model by adding a new vertex inside the loop, thereby breaking the loop into two edges joining the same pair of vertices.} 
Denote by $m$ the number of edges. 
Each edge in $E$ can be viewed as an interval of the cake. 
For any points $x,y$ on an edge, we write $[x,y]$ or $[y,x]$ to denote the interval of the cake between $x$ and $y$; we sometimes identify a vertex $v\in V$ with the corresponding endpoint of the edges adjacent to $v$.  

A \emph{piece of cake} is a finite union of disjoint intervals, where each interval is a subinterval of an edge and the intervals in the piece may belong to different edges.\footnote{As is commonly done in cake cutting, we assume that all intervals are closed intervals. 
Two intervals are said to be \emph{disjoint} if they intersect in at most one point, and two pieces of cake are said to be \emph{disjoint} if they intersect in a finite number of points. 
If we adopt the stricter convention that each point can only be allocated to one agent (so intervals can be open, half-open, or closed) and two intervals are disjoint only if their intersection is empty, there are strong negative results. 
For example, on a star graph, at most one agent would be able to receive intervals from more than one branch in a connected allocation.} 
A piece of cake is said to be \emph{connected} if for any two points in the piece, it is possible to get from one point to the other along the graph $G$ by only traversing this piece of cake. 
Each agent $i$ has a nonnegative \emph{valuation function} (or \emph{utility function}) $f_i$, which specifies the agent's value for each piece of cake. 
An \emph{instance} consists of the graph $G$, the agents, and their valuation functions. As is standard in the cake-cutting literature \citep{Procaccia16}, we assume that the valuation functions are 
\begin{itemize}
    \item \emph{normalized}: the value of an agent for the entire cake is $1$;
    \item \emph{divisible}: for each interval $[x,y]$ and $0\leq\lambda\leq 1$, there is a point $z\in [x,y]$ such that $f_i([x,z])=\lambda\cdot f_i([x,y])$;
    \item \emph{additive}: the value of an agent for a piece of cake is the sum of her values for the intervals in the piece. 
\end{itemize}

An \emph{allocation} of the cake is denoted by a vector $A=(A_1,\dots,A_n)$, where each $A_i$ is a piece of cake, and $A_i$ and $A_j$ are disjoint for all $i\neq j$. 
An allocation is said to be \emph{complete} if the entire cake is allocated, and \emph{connected} if each $A_i$ is a connected piece of cake.
The \emph{egalitarian welfare} of an allocation is defined as $\min_{i\in N}f_i(A_i)$. An allocation is \emph{proportional} if its egalitarian welfare is at least $1/n$. 
The \emph{inequity} of an allocation is defined as $\max_{i,j\in N}|f_i(A_i)-f_j(A_j)|$; an allocation with inequity $0$ is said to be \emph{equitable}.

We make analogous assumptions for chore division (Section~\ref{sec:choredivision}). 
Each agent has a nonnegative \emph{cost function} $f_i$ for the chore, which is normalized, divisible, and additive. 
The \emph{egalitarian cost} of an allocation is defined as $\max_{i\in N} f_i(A_i)$. 
Naturally, we require the entire chore to be allocated, so we restrict our attention to complete allocations of the chore.

\section{Any Number of Agents}
\label{sec:anyagents}

In this section, we present an egalitarian welfare guarantee that holds for any number of agents and arbitrary graphs, and derive improved guarantees in the case where the graph is a star. 

\subsection{General Guarantees}

We begin by showing that it is always possible to give every agent a utility of at least $\frac{1}{2n-1}$, and this bound is tight.
Similarly to the Dubins-Spanier protocol, our algorithm proceeds by identifying a piece that is valuable enough for one agent but at the same time not too valuable for the other agents, allocating such a piece to the former agent, and recursing on the latter agents.

\begin{theorem}
\label{thm:any-agent-egal}
For any graph $G$, there exists a connected allocation with egalitarian welfare at least $\frac{1}{2n-1}$. 
On the other hand, there exists a graph $G$ and identical valuations of the agents such that any connected allocation yields egalitarian welfare at most $\frac{1}{2n-1}$.
\end{theorem}

\begin{proof}
Let $\alpha:=\frac{1}{2n-1}$. To show the second part of the theorem, let $G$ be a star with $2n-1$ edges such that every agent values each edge exactly $\alpha$, and the value is distributed uniformly within the edge. Assume for contradiction that there is a connected allocation with egalitarian welfare strictly greater than $\alpha$. 
Consider any agent $i$. The agent must receive intervals from at least two edges, and these intervals must be connected via the center vertex. 
Note that the unallocated parts of these edges cannot be allocated to other agents, since any agent who receives an interval from such a part cannot receive intervals from other edges and would therefore obtain value less than $\alpha$. 
Hence, at least two edges are only allocated to agent $i$. However, this means that there must be at least $2n$ edges in total, a contradiction.

We now prove the first part of the theorem. Let $G$ be an arbitrary graph. We will show that there exists a moving-knife algorithm that produces a connected allocation with egalitarian welfare at least $\alpha$. 
We proceed by induction on $n$; the statement trivially holds for $n=1$ since we can simply allocate the entire cake to the only agent. Assume that the statement holds for $n-1$ agents, and consider an instance with $n$ agents.

First, we claim that we can turn $G$ into a tree. As long as $G$ contains at least one cycle, pick an edge $uv$ that belongs to a cycle, add a new vertex $v'$, and replace this edge by an edge $uv'$ while keeping the remaining edges of the graph as before.
Since at least one cycle is removed by this operation and no new cycle is created, $G$ eventually becomes a tree. 
Note that any connected allocation of the modified graph is a connected allocation in the original graph with the same value for every agent, so it suffices to prove the theorem for the modified graph.

Choose an arbitrary vertex $u$ of the tree $G$ as its root. Let $v$ be a vertex such that the subtree rooted at $v$ yields value at least $\alpha$ to some agent, and the same does not hold for the subtree rooted at any child of $v$. Let $w_1,\dots,w_k$ be the children of $v$. We consider two cases:

\begin{itemize}
    \item Case 1: At least one of the $k$ branches of $v$ along with the corresponding subtree yields value at least $\alpha$ to some agent. Assume without loss of generality that the branch containing $w_1$ is one such branch. 
    By our assumption, the subtree rooted at $w_1$ yields value less than $\alpha$ to all agents. Hence, by moving a knife from $w_1$ to $v$, we can find the point $x$ closest to $w_1$ such that some agent $i$ values the interval $[w_1,x]$ together with the subtree rooted at $w_1$ exactly $\alpha$, and all other agents value this piece of cake at most $\alpha$.
    We allocate this piece of cake to agent $i$, and make $x$ a new vertex in the remaining graph, which has value at least $1-\alpha$ for each of the remaining agents. 
    By the inductive hypothesis, there exists a connected allocation of the remaining graph to the $n-1$ agents such that every agent receives value at least $\frac{1}{2n-3}\cdot(1-\alpha)=\frac{1}{2n-3}\cdot\frac{2n-2}{2n-1} > \alpha$, as desired. 
    \item Case 2: Every branch of $v$ along with the corresponding subtree yields value less than $\alpha$ to all agents. 
    Let $t\in\{1,2,\dots,k\}$ be the smallest number such that the first $t$ branches and their subtrees together yield value at least $\alpha$ to some agent $i$. 
    We allocate this piece of cake to agent $i$. For every other agent, the first $t-1$ branches is worth less than $\alpha$ and the $t$th branch is also worth less than $\alpha$, so the piece of cake allocated to agent $i$ is worth less than $2\alpha$. 
    By the inductive hypothesis, there exists a connected allocation of the remaining graph to the $n-1$ agents such that every agent receives value at least $\frac{1}{2n-3}\cdot(1-2\alpha) = \frac{1}{2n-1} = \alpha$, as desired.
\end{itemize}
The two cases together complete the induction.
\end{proof}

Note that by following the algorithm in the proof of Theorem~\ref{thm:any-agent-egal}, we also obtain the following statement, which will be useful for our later results.

\begin{lemma}
\label{lem:extract-piece}
Let $H$ be a connected piece of cake in a graph $G$, and suppose that all agents have value $x$ for $H$.
For any $\alpha\leq x$, there exists a partition of $H$ into two connected pieces such that one of the agents has value at least $\alpha$ for the first piece, while all of the remaining agents have value at most $2\alpha$ for this piece.
\end{lemma}

\subsection{Stars}

For certain graphs, it is possible to improve upon the guarantee provided by Theorem~\ref{thm:any-agent-egal}---an obvious example is the graph consisting of a single edge, for which the Dubins-Spanier protocol yields an egalitarian welfare of at least $1/n$. We now derive the optimal egalitarian welfare guarantee in the case where the graph is a star. For integers $n\geq 2$ and $k\geq 3$, define
\[
f(n,k) = 
  \begin{cases}
    \dfrac{1}{n+\lceil k/2\rceil - 1} & \text{for } k < 2n-1;\\
        \dfrac{1}{2n-1} & \text{for } k \geq 2n-1.
  \end{cases}
\]

\begin{theorem}
\label{thm:any-agent-stars}
Let $n\geq 2$ and $k\geq 3$, and let $G$ be a star with $k$ edges. There exists a connected allocation with egalitarian welfare at least $f(n,k)$. 
Moreover, the bound $f(n,k)$ is tight.
\end{theorem}

The values of $f(n,k)$ for small $n$ and $k$ are shown in Table~\ref{table:stars}.

\begin{table}[!ht]
\centering
\begin{tabular}{ c|ccccccccc }
  $n\downarrow\mid k\rightarrow$ & $3$ & $4$ & $5$ & $6$ & $7$ & $8$ & $9$ & $10$ & $11$ \\
  \hline
  $2$ & $1/3$ & $1/3$ & $1/3$ & $1/3$ & $1/3$ & $1/3$ & $1/3$ & $1/3$ & $1/3$ \\
  $3$ & $1/4$ & $1/4$ & $1/5$ & $1/5$ & $1/5$ & $1/5$ & $1/5$ & $1/5$ & $1/5$ \\
  $4$ & $1/5$ & $1/5$ & $1/6$ & $1/6$ & $1/7$ & $1/7$ & $1/7$ & $1/7$ & $1/7$ \\
  $5$ & $1/6$ & $1/6$ & $1/7$ & $1/7$ & $1/8$ & $1/8$ & $1/9$ & $1/9$ & $1/9$ \\
  $6$ & $1/7$ & $1/7$ & $1/8$ & $1/8$ & $1/9$ & $1/9$ & $1/10$ & $1/10$ & $1/11$ \\
  $7$ & $1/8$ & $1/8$ & $1/9$ & $1/9$ & $1/10$ & $1/10$ & $1/11$ & $1/11$ & $1/12$ \\
  $8$ & $1/9$ & $1/9$ & $1/10$ & $1/10$ & $1/11$ & $1/11$ & $1/12$ & $1/12$ & $1/13$ \\
  $9$ & $1/10$ & $1/10$ & $1/11$ & $1/11$ & $1/12$ & $1/12$ & $1/13$ & $1/13$ & $1/14$ \\
  $10$ & $1/11$ & $1/11$ & $1/12$ & $1/12$ & $1/13$ & $1/13$ & $1/14$ & $1/14$ & $1/15$ \\
\end{tabular}
\caption{Some values of $f(n,k)$, the best egalitarian welfare guarantee for $n$ agents and a star with $k$ edges.}
\label{table:stars}
\end{table}

\begin{proof}
We proceed by induction on $n$. For the base case $n=2$, the lower bound of $1/3$ follows from Theorem~\ref{thm:any-agent-egal}. 
To see that $1/3$ is also an upper bound, assume that both agents value three of the edges uniformly at exactly $1/3$ and have no value for the remaining $k-3$ edges. 
Any connected allocation with egalitarian welfare greater than $1/3$ would give rise to a connected allocation of a three-edge star with egalitarian welfare greater than $1/3$, which by Theorem~\ref{thm:any-agent-egal} does not exist.

Assume that the statement holds for $n-1$ agents. First, we show that there exists a connected allocation with egalitarian welfare at least $f(n,k)$. 
If $k\geq 2n-1$, this follows immediately from Theorem~\ref{thm:any-agent-egal}. Suppose that $k\leq 2n-2$. We have $n+\lceil k/2\rceil -1 \geq n + k/2 - 1 \geq \frac{k+2}{2}+\frac{k}{2}-1 = k$, or $f(n,k) \leq 1/k$. 
Hence, every agent has value at least $f(n,k)$ for some edge of the star. We choose one such edge for an arbitrary agent and move a knife from its outer endpoint to the center of the star, stopping when the covered part has value $f(n,k)$ for some agent $i$. 
We allocate this piece of cake to agent $i$, and make the cut point a new vertex in the remaining graph, which has value at least $1-f(n,k)$ for each of the remaining agents. 
The remaining graph is still a star with $k$ edges (possibly with a degenerate edge), so by the inductive hypothesis, there exists a connected allocation of the remaining graph to the $n-1$ agents with egalitarian welfare at least $$f(n-1,k)\cdot(1-f(n,k)) = \frac{1}{n+\lceil k/2\rceil -2}\cdot \frac{n+\lceil k/2\rceil -2}{n+\lceil k/2\rceil -1} = \frac{1}{n+\lceil k/2\rceil -1},$$ 
where the first equality follows from the observation that $f(n-1,k) = \frac{1}{n+\lceil k/2\rceil -2}$ for all $k\leq 2n-2$.

Next, we show that $f(n,k)$ is tight for all $n$ and $k$. If $k\geq 2n-1$, this follows from the instance in Theorem~\ref{thm:any-agent-egal} and by adding extra edges of zero value. 
Suppose that $k\leq 2n-2$ and that the agents have identical valuations. 
Each agent has value $f(n,k)$ for $k-1$ of the edges and $1-(k-1)\cdot f(n,k)$ for the $k$th edge, and the values are distributed uniformly across each edge. 
Note that as in the preceding paragraph, we have $f(n,k)\leq 1/k$, and therefore $1-(k-1)\cdot f(n,k)\geq f(n,k)$. 
Assume for contradiction that there is a connected allocation with egalitarian welfare strictly greater than $f(n,k)$. 
Any agent must either receive intervals from at least two edges (perhaps including the $k$th edge), or only receive an interval from the $k$th edge. 
The number of agents of the first type is at most $\lfloor k/2\rfloor$. 
The number of agents of the second type is strictly less than 
$$\frac{1-(k-1)f(n,k)}{f(n,k)} = \frac{\frac{n-\lfloor k/2\rfloor}{n+\lceil k/2\rceil -1}}{\frac{1}{n+\lceil k/2\rceil -1}} = n-\lfloor k/2\rfloor.$$
This means that the total number of agents is strictly less than $\lfloor k/2\rfloor + (n-\lfloor k/2\rfloor) = n$, yielding the desired contradiction.
\end{proof}

\section{Two Agents}
\label{sec:twoagents}

In this section, we focus on the case of two agents. We establish the optimal egalitarian welfare that can be obtained for each graph and derive utility frontiers when the agents may have different entitlements. In addition, we explore the extent to which our guarantees can be improved if we allow more than one connected piece per agent, and also consider approximate equitability.

\subsection{Graph-Specific Guarantees}

Before we can state our results for specific graphs, we need some graph-theoretic terminology.
Recall that a \emph{bridge} of a graph is an edge that is not contained in any cycle. A graph is said to be \emph{bridgeless} if it contains no bridges. 

\begin{definition}
A graph is said to be \emph{almost bridgeless} if we can add an edge so that the resulting graph is bridgeless.
\end{definition}

Note that according to this definition, every connected bridgeless graph with at least two vertices is also almost bridgeless, since we can add a copy of an existing edge.

Next, we define an oriented labeling of a graph.

\begin{definition}
An \emph{oriented labeling} of a graph with $m$ edges is a labeling of the edges with numbers $1,2,\dots,m$, using each number exactly once, together with a labeling of one endpoint of each edge $i$ with $i^-$ and the other endpoint with $i^+$ (so each vertex receives a number of labels equal to the number of edges adjacent to it). 
An oriented labeling is said to be \emph{contiguous} if:
\begin{itemize}
    \item For each $2\leq i\leq m$, the edges labeled $1,2,\dots,i-1$ form a connected subgraph, and the vertex labeled $i^-$ belongs to one of these edges.
    \item For each $1\leq i\leq m-1$, the edges labeled $i+1,i+2,\dots,m$ form a connected subgraph, and the vertex labeled $i^+$ belongs to one of these edges.
\end{itemize}
\end{definition}

It turns out that a graph admitting a contiguous oriented labeling is equivalent to it being almost bridgeless.

\begin{lemma}
\label{lem:almost-bridgeless}
A graph is almost bridgeless if and only if it admits a contiguous oriented labeling.
\end{lemma}

\begin{proof}
($\Leftarrow$) Assume that a graph admits a contiguous oriented labeling. Add an edge between the vertices labeled $1^-$ and $m^+$. We claim that the resulting graph is bridgeless. 
Since the vertices $1^-$ and $m^+$ are connected in the original graph, the new edge is part of a cycle. Now, consider any edge in the original graph; assume that the edge has label $i$. 
The vertices $1^-$ and $i^-$ are connected by a path that only goes through edges between $1$ and $i-1$. Likewise, the vertices $i^+$ and $m^+$ are connected by a path that only goes through edges between $i+1$ and $m$. 
Hence, the edge $i$ belongs to a cycle that goes through the two paths, the edge between $m^+$ and $1^-$, and itself. 

($\Rightarrow$) Assume that a graph is almost bridgeless. We will label all edges with labels $1,2,\dots,m$ and orient each edge in one direction (the source of edge $i$ corresponds to the vertex $i^-$ and the sink to the vertex $i^+$) so that the labeling is a contiguous oriented labeling. 

Suppose that the graph becomes bridgeless if we add an edge $uv$. Consider a path from $u$ to $v$, and sort the edges and orient them along this path. We will iteratively construct \emph{ears} until all edges are used. 
Each ear is a path starting at a vertex of a previous ear and ending at a vertex of a previous ear (possibly the same as the former vertex, in which case the path becomes a cycle) but not going through any other vertex of a previous ear; the only exception is the first ear, which is the path from $u$ to $v$. 
Suppose that we have constructed some ears, and not all edges have been used. If there are edges that connect only vertices in the existing ears, we make each such edge into a new ear. 
If some edges still remain after this process, then since the graph is connected, there must be an edge $xy$ such that $x$ belongs to an existing ear but $y$ does not. 
By assumption, $xy$ is contained in a cycle if the edge $uv$ is added, so we can follow the edges in this cycle until we reach a vertex $z$ in an existing ear for the first time (possibly $z=x$). 
Since we stop if we reach either $u$ or $v$, edge $uv$ cannot be part of this trail, so we have a new ear. Assume without loss of generality that either $x=u$, or the first edge directed into $x$ appears no later than the first edge directed into $z$ in the current edge order, or both. Orient the edges of the new ear along the trail from $x$ to $z$. 
If $x=u$, insert these edges consecutively at the beginning of the order. Else, insert them consecutively after the first edge directed into $x$. 
After all edges have been added, label them from $1$ to $m$ according to the final order. Observe that the edge with label $1$ is always adjacent to $u$, and the edge with label $m$ is always adjacent to $v$.

We claim that the resulting oriented labeling is contiguous. First, we show that for each $1\leq i\leq m$, the edges labeled $1,2,\dots,i$ form a connected subgraph. 
We proceed by induction on $i$, with the base case $i=1$ holding trivially. Assume that the edges $1,2,\dots,i-1$ form a connected subgraph for some $i\geq 2$. 
If edge $i$ is not the first edge in its ear, its predecessor in its ear has label at most $i-1$, so the induction hypothesis implies that the edges $1,2,\dots,i$ form a connected subgraph. 
Moreover, in this case, vertex $i^-$ belongs to the predecessor edge with label at most $i-1$. Suppose now that edge $i$ is the first edge in its ear, and assume that the edge is directed from $x$ to $y$. 
If $x=u$, then since $i\geq 2$, the edge with label $1$ is an edge with a lower label that is adjacent to $u$. Else, the first edge directed into $x$ has label at most $i-1$. 
In either case, the induction hypothesis implies that the edges $1,2,\dots,i$ form a connected subgraph, and vertex $i^-=x$ is adjacent to a vertex with label at most $i-1$. 
This completes the induction and moreover shows that for each $2\leq i\leq m$, vertex $i^-$ belongs to one of the edges $1,2,\dots,i-1$. 

Next, we show that for each $1\leq i\leq m$, the edges labeled $i,i+1,\dots,m$ form a connected subgraph. 
We proceed by downward induction on $i$, with the base case $i=m$ holding trivially. Assume that the edges $i+1,i+2,\dots,m$ form a connected subgraph for some $i<m$. 
If edge $i$ is not the last edge in its ear, its successor in its ear has label at least $i+1$, so the induction hypothesis implies that the edges $i,i+1,\dots,m$ form a connected subgraph.
Moreover, in this case, vertex $i^+$ belongs to the successor edge with label at least $i+1$. Suppose now that edge $i$ is the last edge in its ear, and assume that the edge is directed from $y$ to $z$. 
If $z=v$, then since $i<m$, the edge with label $m$ is an edge with a higher label that is adjacent to $v$. Suppose that $z\neq v$, which also implies that $i$ does not belong to the first ear. 
Consider the moment when we insert the ear that contains $i$, and assume that this ear begins at vertex $x$. 
If $x=u$, then the edges of this ear are placed at the beginning of the order; since $z$ appears in a previous ear, there is an edge adjacent to it that comes after edge $yz$ in the order. 
Suppose therefore that $x\neq u$, so that at the moment before we insert the ear containing $i$, the first edge into $x$ appears no later than the first edge directed into $z$ in the order. We place the edges of the new ear after the first edge into $x$. 
If $x\neq z$, there is an edge adjacent to $z$ that comes after this ear in the order. Else, $x=z$, and this vertex is different from $u$ and $v$. 
Consider the earliest ear that contains $z$, and note that in this ear, there is an edge into $z$ and another edge out of $z$; let $i_1$ and $i_2$ be the label of the two edges respectively. 
The first edge into $z$ appears no later than $i_1$, so the ear containing $i$ appears before $i_2$ in the ordering. Hence, there is an edge with label at least $i+1$ adjacent to $z$. This completes the induction. 
It also follows from our argument that for each $1\leq i\leq m-1$, the vertex $i^+$ belongs to one of the edges $i+1,i+2,\dots,m$. Combined with the previous paragraph, we find that our oriented labeling is contiguous, as claimed.
\end{proof}

Given a graph $G$ with $k$ vertices, a \emph{bipolar numbering} of $G$ is a labeling of the \emph{vertices} with numbers $1,2,\dots,k$, with each number used exactly once, such that every vertex with label greater than $1$ has a neighbor with a smaller label and each vertex with label smaller than $k$ has a neighbor with a larger label.
Bil\`{o} et al.~\citet{BiloCaFl19} characterized the class of graphs that admit a bipolar numbering as the graphs with the property that if the vertices of the graph represent indivisible items of possibly different values to the two agents, there always exists a connected `envy-free up to one item' allocation. 
We show that the class of graphs that admit a bipolar numbering forms a strict subclass of the almost bridgeless graphs (which, by Lemma~\ref{lem:almost-bridgeless}, is equivalent to the class of graphs that admit a contiguous oriented labeling).

\begin{proposition}
Any graph that admits a bipolar numbering is almost bridgeless, but the converse does not hold.
\end{proposition}

\begin{proof}
Assume that a graph admits a bipolar numbering with the vertices labeled $1,2,\dots,k$. 
This means that for any vertex $i$, there exists a path from $1$ to $i$ that only goes through vertices $1,2,\dots,i$, and a path from $i$ to $k$ that only goes through vertices $i,i+1,\dots,k$. 
Add an edge between vertices $1$ and $k$. We will show that the resulting graph is bridgeless, i.e., every edge is contained in a cycle. This is clear for the new edge. 
For any edge in the original graph between vertices $i$ and $j$ with $i<j$, we can construct a cycle containing it by following a path from $j$ to $k$ that only uses vertices $j,j+1,\dots,k$, traversing the added edge from $k$ to $1$, and following a path from $1$ to $i$ that only uses vertices $1,2,\dots,i$.

To show that the converse does not hold, consider the graphs shown in Figure~\ref{fig:flowers}. 
Since every edge in both graphs is contained in a cycle, both graphs are bridgeless and therefore almost bridgeless. On the other hand, one can check that neither graph admits a bipolar numbering.
\end{proof}

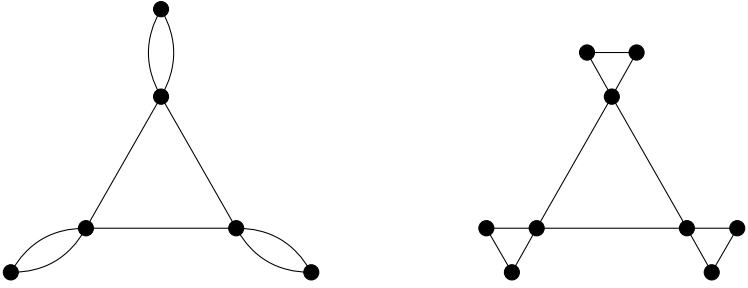
\begin{figure}[!ht]
\centering
\begin{tikzpicture}
\draw (1,0.58) -- (3,0.58) -- (2,2.31) -- (1,0.58);
\draw [fill] (1,0.58) circle [radius = 0.1];
\draw [fill] (3,0.58) circle [radius = 0.1];
\draw [fill] (2,2.31) circle [radius = 0.1];
\draw [fill] (0,0) circle [radius = 0.1];
\draw [fill] (4,0) circle [radius = 0.1];
\draw [fill] (2,3.46) circle [radius = 0.1];
\draw (1,0.58) to[bend right] (0,0);
\draw (1,0.58) to[bend left] (0,0);
\draw (3,0.58) to[bend right] (4,0);
\draw (3,0.58) to[bend left] (4,0);
\draw (2,2.31) to[bend right] (2,3.46);
\draw (2,2.31) to[bend left] (2,3.46);

\draw (7,0.58) -- (9,0.58) -- (8,2.31) -- (7,0.58);
\draw [fill] (7,0.58) circle [radius = 0.1];
\draw [fill] (9,0.58) circle [radius = 0.1];
\draw [fill] (8,2.31) circle [radius = 0.1];
\draw [fill] (7.67,2.89) circle [radius = 0.1];
\draw [fill] (8.33,2.89) circle [radius = 0.1];
\draw (8,2.31) -- (7.67,2.89) -- (8.33,2.89) -- (8,2.31);
\draw [fill] (6.33,0.58) circle [radius = 0.1];
\draw [fill] (6.67,0) circle [radius = 0.1];
\draw (7,0.58) -- (6.33,0.58) -- (6.67,0) -- (7,0.58);
\draw [fill] (9.67,0.58) circle [radius = 0.1];
\draw [fill] (9.33,0) circle [radius = 0.1];
\draw (9,0.58) -- (9.67,0.58) -- (9.33,0) -- (9,0.58);
\end{tikzpicture}
\caption{Examples of graphs that are almost bridgeless but do not admit a bipolar numbering.}
\label{fig:flowers}
\end{figure}

We are now ready to show our classification result: the optimal egalitarian welfare that can always be obtained for a graph is $1/2$ if the graph is almost bridgeless, and $1/3$ otherwise. The former is shown in Theorem~\ref{thm:two-agent-1/2}, while the latter follows from Theorems~\ref{thm:any-agent-egal} and \ref{thm:two-agent-1/3}.

\begin{theorem}
\label{thm:two-agent-1/2}
For $n=2$ and any almost bridgeless graph $G$, there exists a connected proportional allocation.
\end{theorem}

\begin{proof}
Suppose that $G$ is almost bridgeless. By Lemma~\ref{lem:almost-bridgeless}, it admits a contiguous oriented labeling. 
We move a knife over the edges of $G$ in increasing order of the label. For each edge with label $i$, the knife goes from vertex $i^-$ to vertex $i^+$. 
If the knife is currently on edge $i$, we stop when the piece of cake containing edges $1,2,\dots,i-1$, together with the interval of edge $i$ between $i^-$ and the current position of the knife, yields value exactly $1/2$ to one of the agents.
We allocate this piece of cake to the agent who receives value $1/2$, and the remainder of the cake to the other agent. 
Both agents receive value at least $1/2$ and, by definition of the labeling, obtain a connected allocation of the cake.
\end{proof}

\begin{lemma}
\label{lem:threebridges}
Let $F$ be a nonempty set of bridges in a graph. If no path contains all bridges in $F$, there exist three bridges in $F$ such that no path contains all three bridges.
\end{lemma}

\begin{proof}
Assume that no path contains all bridges in $F$. Since bridges are not contained in cycles, there is a spanning tree $T$ that contains $F$. Let $T'$ be a minimal subtree of $T$ that contains $F$. 
Since $T'$ cannot be a path, it has a vertex $v$ with degree at least $3$. Each of the (at least three) branches of $v$ must contain a bridge from $F$---otherwise the branch can be removed to obtain a smaller subtree than $T'$. 

Let $e_1,e_2,e_3$ be bridges contained in three distinct branches. Suppose for contradiction that they are contained in a path. 
By reversing the direction of the path if necessary, we may assume that the path traverses at least two of the edges away from $v$ with respect to $T'$. 
Assume further that two of these edges are $e_1$ and $e_2$, and that the path traverses $e_1$ before $e_2$. 
After traversing $e_1$, the path must reach another vertex $w$ in $T'$ that lies on the same side as $v$ with respect to both $e_1$ and $e_2$ (possibly the endpoint of $e_1$ or $e_2$ closer to $v$). 
By combining the portion of the path from $e_1$ to $w$ with the path in $T'$ from $w$ to $e_1$, we find that $e_1$ lies on a cycle, a contradiction.
\end{proof}

\begin{theorem}
\label{thm:two-agent-1/3}
For $n=2$ and any graph $G$ that is not almost bridgeless, there exist identical valuation functions of the two agents such that any connected allocation yields egalitarian welfare at most $1/3$.
\end{theorem}

\begin{proof}
Suppose that $G$ is not almost bridgeless. Then no path can contain all bridges of $G$: if there exists such a path, we can eliminate all bridges by adding an edge that connects the endpoints of this path.
By Lemma~\ref{lem:threebridges}, there exist three bridges of $G$ such that no path contains all three bridges. 
For each of the three bridges, the other two bridges must lie on the same side of it, since otherwise we can construct a path that contains all three bridges. 

Assume that both agents value each of the three bridges exactly $1/3$ (and every other edge $0$), and the value is distributed uniformly within each bridge. 
Suppose for contradiction that there exists a connected allocation with egalitarian welfare strictly greater than $1/3$. 
This means that each agent must receive intervals from at least two bridges. However, when an agent receives intervals from two bridges, each interval must contain the endpoint of the bridge that is on the same side as the other two bridges. 
This is impossible since there are only three bridges, yielding the desired contradiction.
\end{proof}

\subsection{Utility Frontiers}

In this subsection, we establish the frontiers of the utilities that we can guarantee to the two agents regardless of the graph, assuming that the agents may have different entitlements. 
We begin by observing that the cut-and-choose protocol allows us to find an allocation that gives utility $1/2$ to the first agent and $1/3$ to the second agent; this generalizes the case $n=2$ of Theorem~\ref{thm:any-agent-egal}.

\begin{theorem}
\label{thm:two-agent-fixed}
For $n=2$ and any graph $G$, there exists a connected allocation such that the first agent receives value at least $1/2$ and the second agent receives value at least $1/3$. 
\end{theorem}

\begin{proof}
By Theorem~\ref{thm:any-agent-egal}, there exists a partition of the cake into two connected pieces such that the second agent values both pieces at least $1/3$.
The first agent can then simply choose the piece that she prefers and obtain value at least $1/2$.
\end{proof}

The next proposition follows from Theorem~\ref{thm:any-agent-egal}.

\begin{proposition}
\label{prop:two-agent-1/3}
For $n=2$, there exists a graph $G$ and identical valuations of the two agents such that any connected allocation yields egalitarian welfare at most $1/3$.
\end{proposition}

To complete the utility frontier, we show that if we are required to give a utility of more than $1/2$ to the first agent, it may be impossible to provide any nontrivial guarantee for the second agent.

\begin{proposition}
\label{prop:two-agent-1/2}
Let $\alpha>1/2$ and $\beta>0$. There exists an instance with $n=2$ such that no connected allocation yields value at least $\alpha$ to the first agent and at least $\beta$ to the second agent.
\end{proposition}

\begin{proof}
Fix $\alpha>1/2$ and $\beta>0$, and assume that $G$ consists of a single edge represented by the interval $[0,1]$. Suppose that the first agent values the entire interval $[0,1]$ uniformly, while the second agent values the interval $[1-\alpha,\alpha]$ uniformly and nothing else. 
In any connected allocation that yields value at least $\alpha$ to the first agent, this agent must receive the entire interval $[1-\alpha,\alpha]$. However, that means the second agent receives value $0$ from the allocation.
\end{proof}

Combining Theorem~\ref{thm:two-agent-fixed} with Propositions~\ref{prop:two-agent-1/3} and \ref{prop:two-agent-1/2}, we find that the values of $(\alpha,\beta)$ for which a connected allocation that yields utility $\alpha$ to the first agent and $\beta$ to the second agent always exists are as shown in Figure~\ref{fig:two-agent-fixed}.

\begin{figure}[!ht]
\centering
\begin{tikzpicture}[scale=5]
\draw (0,1) -- (1,1) -- (1,0);
\draw [thick,red] (0,1) -- (0,0) -- (1,0);
\node [below left] at (0,0) {0};
\node [below] at (1,0) {1};
\node [below] at (1/2,0) {1/2};
\node [below] at (1/3,0) {1/3};
\node at (0.5,-0.2) {$\alpha$};
\node at (-0.25,0.5) {$\beta$};
\node [left] at (0,1) {1};
\node [left] at (0,1/2) {1/2};
\node [left] at (0,1/3) {1/3};
\draw [fill=red,thick,red] (0,0) rectangle (1/2,1/3);
\draw [fill=red,thick,red] (0,0) rectangle (1/3,1/2);
\end{tikzpicture}
\caption{The red region corresponds to the values of $(\alpha,\beta)$ such that in any instance with two agents, there exists a connected allocation that yields utility $\alpha$ to the first agent and $\beta$ to the second agent.}
\label{fig:two-agent-fixed}
\end{figure}
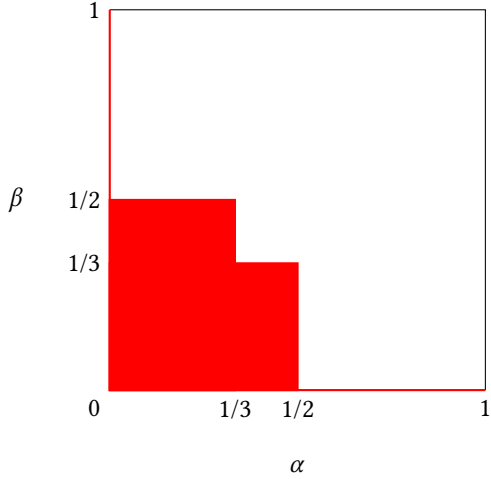

While Figure~1 completely captures the shares that can be guaranteed to the two agents, if we do not fix the entitlements of the agents in advance, it is possible to achieve better guarantees. 
As an example, consider the case where the graph $G$ consists of a single edge. In this case, Proposition~\ref{prop:two-agent-1/2} shows that any entitlements $(\alpha,\beta)$ with $\alpha>1/2$ and $\beta>0$ cannot be achieved. 
On the other hand, for any $\alpha\in[0,1]$, it is possible to give one agent a value of at least $\alpha$ and the other agent a value of at least $1-\alpha$ (while not fixing which agent receives which share). 
Indeed, if we run a moving knife on the single edge and stop when the part already covered by the knife yields value $\alpha$ to some agent, the two pieces can be allocated to yield the desired guarantee. 
In what follows, we determine all shares $(\alpha,\beta)$ for which there always exists a connected allocation that yields value $\alpha$ to one agent and $\beta$ to the other agent regardless of the graph.
Note that Theorem~\ref{thm:two-agent-fixed} carries over, and so does Proposition~\ref{prop:two-agent-1/3} since it holds for agents with identical valuations. 
We fill in the utility frontier, starting with the negative results.

\begin{proposition}
\label{prop:two-agent-1/2-1/4}
Let $\alpha>1/2$ and $\beta>1/4$. There exists an instance with $n=2$ and identical valuations such that no connected allocation yields value at least $\alpha$ to one agent and at least $\beta$ to the other agent.
\end{proposition}

\begin{proof}
Let $G$ be a star with four edges such that every agent values each edge exactly $1/4$, and the value is distributed uniformly within the edge. 
Assume for contradiction that there is a connected allocation that yields value at least $\alpha$ to one agent and at least $\beta$ to the other agent; since the agents have identical valuations, we may assume without loss of generality that these are agents 1 and 2 respectively. 
Agent 1 must receive intervals from at least three edges, and these intervals must be connected via the center vertex. Note that the unallocated parts of these edges cannot be allocated to agent 2, since agent 2 would receive value less than $1/4$. 
Similarly, agent 2 must receive intervals from at least two edges, and these intervals, which again must be connected via the center vertex, cannot be allocated to agent 1. 
However, this means that there must be at least five edges in total, a contradiction.
\end{proof}

\begin{figure}[!ht]
\centering
\begin{tikzpicture}
\draw (0,1)--(2,2)--(5,2)--(7,3);
\draw (0,3)--(2,2);
\draw (5,2)--(7,1);
\draw [fill] (0,1) circle [radius = 0.1];
\draw [fill] (2,2) circle [radius = 0.1];
\draw [fill] (5,2) circle [radius = 0.1];
\draw [fill] (7,3) circle [radius = 0.1];
\draw [fill] (0,3) circle [radius = 0.1];
\draw [fill] (7,1) circle [radius = 0.1];
\node at (3.5,2.3) {$1-4\alpha+4\epsilon$};
\node[rotate=30] at (6,2.8) {$\alpha-\epsilon$};
\node[rotate=30] at (1,1.25) {$\alpha-\epsilon$};
\node[rotate=-30] at (6,1.25) {$\alpha-\epsilon$};
\node[rotate=-30] at (1,2.8) {$\alpha-\epsilon$};
\end{tikzpicture}
\caption{Graph $G$ in the proof of Proposition~\ref{prop:two-agent-alpha-negative}.}
\label{fig:two-agent-alpha-negative}
\end{figure}
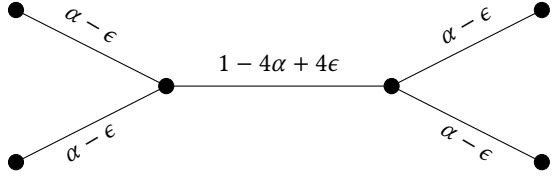

\begin{proposition}
\label{prop:two-agent-alpha-negative}
Let $\alpha\leq 1/4$ and $\beta>1-2\alpha$. 
There exists an instance with $n=2$ and identical valuations such that no connected allocation yields value at least $\alpha$ to one agent and at least $\beta$ to the other agent.
\end{proposition}

\begin{proof}
Choose $\epsilon>0$ such that $\beta>1-2\alpha+2\epsilon$. 
Let $G$ be the graph shown in Figure~\ref{fig:two-agent-alpha-negative}, where the value of each agent for each edge is as shown in the figure and distributed uniformly across the edge. 
Assume for contradiction that there is a connected allocation that yields value at least $\alpha$ to one agent and at least $\beta$ to the other agent; since the agents have identical valuations, we may assume without loss of generality that these are agents 1 and 2 respectively. 
Agent 2 must receive part of the middle edge; otherwise she has value at most $2(\alpha-\epsilon)<1/2<\beta$. 
Moreover, she must receive part of some left edge and part of some right edge; otherwise she has value at most $2(\alpha-\epsilon)+(1-4\alpha+4\epsilon) = 1-2\alpha+2\epsilon<\beta$. Hence, the agent must receive the entire middle edge. 
This means that agent 1 can receive intervals from only one non-middle edge. Her value is therefore at most $\alpha-\epsilon$, a contradiction.
\end{proof}

We now move on to the positive result, which shows the additional guarantee that we can obtain if we give up control over which agent receives which entitlement.

\begin{theorem}
\label{thm:two-agent-alpha-positive}
Let $\alpha\leq 1/4$. For $n=2$ and any graph $G$, there exists a connected allocation such that one agent receives value at least $\alpha$ and the other agent receives value at least $1-2\alpha$. 
\end{theorem}

\begin{proof}
This follows immediately from Lemma~\ref{lem:extract-piece} by taking $H$ to be the entire graph $G$.
\end{proof}

Combining Theorem~\ref{thm:two-agent-alpha-positive} with Propositions~\ref{prop:two-agent-1/2-1/4} and \ref{prop:two-agent-alpha-negative}, we find that the values of $(\alpha,\beta)$ for which a connected allocation that yields utility $\alpha$ to one agent and $\beta$ to the other agent always exists are as shown in Figure~\ref{fig:two-agent-flexible}.

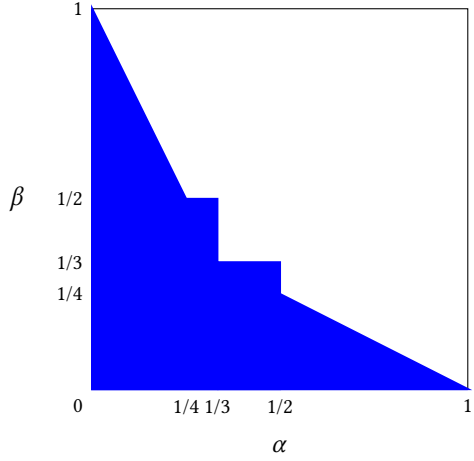
\begin{figure}[!ht]
\centering
\begin{tikzpicture}[scale=5]
\draw (0,1) -- (1,1) -- (1,0);
\draw [thick,blue] (0,1) -- (0,0) -- (1,0);
\node [below left] at (0,0) {\footnotesize 0};
\node [below] at (1,0) {\footnotesize 1};
\node [below] at (1/2,0) {\footnotesize 1/2};
\node [below] at (1/3,0) {\footnotesize 1/3};
\node [below] at (1/4,0) {\footnotesize 1/4};
\node at (0.5,-0.15) {$\alpha$};
\node at (-0.2,0.5) {$\beta$};
\node [left] at (0,1) {\footnotesize 1};
\node [left] at (0,1/2) {\footnotesize 1/2};
\node [left] at (0,1/3) {\footnotesize 1/3};
\node [left] at (0,1/4) {\footnotesize 1/4};
\draw [fill=blue,thick,blue] (0,0) rectangle (1/2,1/3);
\draw [fill=blue,thick,blue] (0,0) rectangle (1/3,1/2);
\draw [fill=blue,thick,blue] (0,1/3) -- (1/3,1/3) -- (0,1) -- (0,1/3);
\draw [fill=blue,thick,blue] (1/3,0) -- (1/3,1/3) -- (1,0) -- (1/3,0);
\end{tikzpicture}
\caption{The blue region corresponds to the values of $(\alpha,\beta)$ such that in any instance with two agents, there exists a connected allocation that yields utility $\alpha$ to one agent and $\beta$ to the other agent.}
\label{fig:two-agent-flexible}
\end{figure}

\subsection{Beyond One Connected Piece}

As we mentioned in the introduction, one important motivation for considering connected allocations is to avoid situations where an agent receives a ``union of crumbs''. 
In light of this motivation, it is interesting to explore whether we can obtain improved guarantees if we allow the agents to receive a small number of connected pieces. 
We demonstrate in this subsection that such improvements are indeed possible by presenting a tight bound of $\frac{1}{2}-\frac{1}{2\cdot 3^k}$ on the egalitarian welfare that can be guaranteed when a total of $k+1$ connected pieces are permitted.
We first establish the lower bound.

\begin{theorem}
\label{thm:two-agent-k-connected-lower}
Let $k$ be a positive integer.
For $n=2$ and any graph $G$, there exists an allocation in which the two agents receive a total of at most $k+1$ connected pieces and the egalitarian welfare is at least $\frac{1}{2}-\frac{1}{2\cdot 3^k}$.
\end{theorem}

\begin{proof}
Let $G$ be an arbitrary graph. 
It suffices to show that there exists a partition of $G$ into two parts with at most $k+1$ connected pieces in total such that both parts yield value at least $\frac{1}{2}-\frac{1}{2\cdot 3^k}$ to the first agent.
Indeed, given such a partition, we can let the second agent choose the part that she prefers and obtain value at least $1/2$. 
We therefore consider only the first agent from now on.

We proceed by induction on $k$; the base case $k=1$ follows from Theorem~\ref{thm:two-agent-fixed}.
Suppose that the statement holds for $k-1$, i.e., there exists a partition of $G$ into two parts with at most $k$ connected pieces in total such that both parts yield value at least $\frac{1}{2}-\frac{1}{2\cdot 3^{k-1}}$ to the agent.
Assume without loss of generality that the second part has value at least $1/2$, so the first part has value $\frac{1}{2}-x$ for some $0\leq x\leq\frac{1}{2\cdot 3^{k-1}}$.
Since the second part consists of at most $k-1$ connected pieces, it contains a connected piece of value at least $\frac{1}{2k-2}$.
Denote this piece by $H$.

Since $2k-2\leq 3^k$, we have $\frac{2x}{3}\leq \frac{1}{3^k}\leq \frac{1}{2k-2}$.
By creating a duplicate of our agent and applying Lemma~\ref{lem:extract-piece}, we can partition $H$ into two connected pieces in such a way that our agent has value in the range $[\frac{2x}{3},\frac{4x}{3}]$ for the first piece.
Move this piece from the second part of our partition of $G$ to the first part.
The resulting partition of $G$ consists of at most $k+1$ connected pieces in total, and the first part of this partition has value in the range $[\frac{1}{2}-\frac{x}{3},\frac{1}{2}+\frac{x}{3}]$.
This implies that both parts of the partition yield value at least $\frac{1}{2}-\frac{x}{3}\geq\frac{1}{2}-\frac{1}{2\cdot 3^k}$, completing the induction.
\end{proof}

Next, we show that the bound established in Theorem~\ref{thm:two-agent-k-connected-lower} is tight for every $k$.
First we need the following technical lemma.

\begin{lemma}
\label{lem:powers-of-three}
Let $t$ be a positive integer, and let $a_1,a_2,\dots,a_t$ be (not necessarily distinct) integers and $\varepsilon_1,\varepsilon_2,\dots,\varepsilon_t\in\{\pm 1, \pm 2\}$. 
Then
$$\left|\varepsilon_1\cdot 3^{a_1}+\varepsilon_2\cdot 3^{a_2}+\dots+\varepsilon_t\cdot 3^{a_t}-\frac{1}{2}\right|\geq\frac{1}{2\cdot 3^t}.$$
\end{lemma}

\begin{proof}
We proceed by strong induction on $t$. The base case $t=1$ follows from the observation that the terms of the form $\varepsilon\cdot 3^a$ closest to $1/2$ are $1/3$ and $2/3$, and $|1/3-1/2|=|2/3-1/2|=1/6$. 
Suppose that the statement holds up to $t-1$, and assume without loss of generality that $a_1\geq a_2\geq\dots\geq a_t$. 
We process the sum $\varepsilon_1\cdot 3^{a_1}+\varepsilon_2\cdot 3^{a_2}+\dots+\varepsilon_t\cdot 3^{a_t}$ from right to left. 
Consider the moment when we process terms involving $3^a$.
If there are two terms involving $3^a$ with coefficients of opposite signs, we either cancel them or combine them into one term.
So we may assume that all terms with $3^a$ have the same sign, say positive.
If there are two terms $1\cdot 3^a$, we combine them into $2\cdot 3^a$; if there is a term $1\cdot 3^a$ and $2\cdot 3^a$, we combine them into $1\cdot 3^{a+1}$; and if there are two terms $2\cdot 3^a$, we replace them with $1\cdot 3^a$ and $1\cdot 3^{a+1}$.
Our operations do not increase the number of terms, so our procedure terminates with at most one term involving each power of $3$.

If the number of terms $\varepsilon\cdot 3^a$ is now less than $t$, we may apply the induction hypothesis and obtain our desired conclusion.
Assume therefore that the number of terms is still $t$, and  $a_1>a_2>\dots>a_t$. 
We have
$$\left|\sum_{i=2}^t \varepsilon_i\cdot 3^{a_i}\right|
\leq \sum_{i=2}^t\left|\varepsilon_i\cdot 3^{a_i}\right| 
\leq 2\cdot \sum_{i=2}^t 3^{a_i} < 2\cdot\sum_{i=2}^\infty 3^{a_i} = 3^{a_2+1} \leq 3^{a_1}. $$
If $\varepsilon_1$ is negative, we have $\sum_{i=1}^t \varepsilon_i\cdot 3^{a_i} < -3^{a_1}+3^{a_1} = 0$, so the desired statement holds. Assume now that $\varepsilon_1$ is positive. 
If $a_1\leq -2$, we have $\sum_{i=1}^t \varepsilon_i\cdot 3^{a_i}\leq 2\cdot 3^{a_1}+3^{a_1} = 3^{a_1+1}\leq 1/3$, so the desired statement holds in this case.

Suppose that $a_1=-1$. If $\varepsilon_1 = 1$, we have
$$\left|\sum_{i=1}^t \varepsilon_i\cdot 3^{a_i}-\frac{1}{2}\right| = \left|\sum_{i=2}^t \varepsilon_i\cdot 3^{a_i}-\frac{1}{6}\right| = \frac{1}{3}\left|\sum_{i=2}^t \varepsilon_i\cdot 3^{a_i+1}-\frac{1}{2}\right|\geq\frac{1}{3}\cdot\frac{1}{2\cdot 3^{t-1}}=\frac{1}{2\cdot 3^t},$$
where the inequality follows from the inductive hypothesis. 
Similarly, if $\varepsilon_1 = 2$, we have
$$\left|\sum_{i=1}^t \varepsilon_i\cdot 3^{a_i}-\frac{1}{2}\right| = \left|\sum_{i=2}^t \varepsilon_i\cdot 3^{a_i}+\frac{1}{6}\right| = \frac{1}{3}\left|\sum_{i=2}^t (-\varepsilon_i)\cdot 3^{a_i+1}-\frac{1}{2}\right|\geq\frac{1}{3}\cdot\frac{1}{2\cdot 3^{t-1}}=\frac{1}{2\cdot 3^t}.$$

Suppose now that $a_1=0$.
If $\varepsilon_1 = 1$, we have
$$\left|\sum_{i=1}^t \varepsilon_i\cdot 3^{a_i}-\frac{1}{2}\right| = \left|\sum_{i=2}^t \varepsilon_i\cdot 3^{a_i}+\frac{1}{2}\right| = \left|\sum_{i=2}^t (-\varepsilon_i)\cdot 3^{a_i}-\frac{1}{2}\right|\geq\frac{1}{2\cdot 3^{t-1}}>\frac{1}{2\cdot 3^t},$$ where the first inequality follows from the inductive hypothesis.
If $\varepsilon_1 = 2$, we have
$\sum_{i=1}^t \varepsilon_i\cdot 3^{a_i} \geq 2 -\frac{2}{3}-\frac{2}{9}-\dots = 1$, so the desired statement holds.

Finally, suppose that $a_1\geq 1$.
If $\varepsilon_1 = 2$, 
we have
$\sum_{i=1}^t \varepsilon_i\cdot 3^{a_i} \geq 3^{a_i}\left(2 -\frac{2}{3}-\frac{2}{9}-\dots\right) \geq 3$, so the desired statement holds.
Assume now that $\varepsilon_1 = 1$. If $a_2 = a_1-1$ and $\varepsilon_2$ is negative, we may combine the terms $1\cdot 3^{a_1}$ and $\varepsilon_2\cdot 3^{a_2}$ into $(3+\varepsilon_2)\cdot 3^{a_2}$ and apply the induction hypothesis.
So we may assume that either $a_2 \leq a_1-2$ or $\varepsilon_2$ is positive.
In either case, we have
$\sum_{i=1}^t \varepsilon_i\cdot 3^{a_i} \geq 3^{a_1}\left(1-\frac{2}{9}-\frac{2}{27}-\dots\right) = 3^{a_1}\cdot\frac{2}{3}\geq 2$, so the desired statement again holds.
\end{proof}

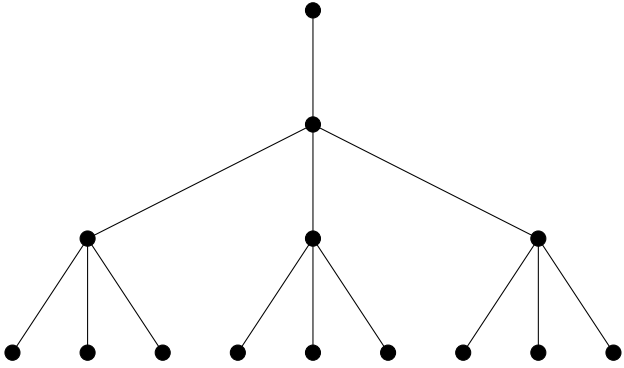
\begin{figure}[!ht]
\centering
\begin{tikzpicture}
\draw (4,4.5) -- (4,3);
\draw (4,3) -- (4,1.5);
\draw (4,3) -- (1,1.5);
\draw (4,3) -- (7,1.5);
\draw (1,1.5) -- (1,0);
\draw (1,1.5) -- (0,0);
\draw (1,1.5) -- (2,0);
\draw (4,1.5) -- (4,0);
\draw (4,1.5) -- (3,0);
\draw (4,1.5) -- (5,0);
\draw (7,1.5) -- (7,0);
\draw (7,1.5) -- (6,0);
\draw (7,1.5) -- (8,0);
\draw [fill] (4,4.5) circle [radius = 0.1];
\draw [fill] (4,3) circle [radius = 0.1];
\draw [fill] (4,1.5) circle [radius = 0.1];
\draw [fill] (1,1.5) circle [radius = 0.1];
\draw [fill] (7,1.5) circle [radius = 0.1];
\draw [fill] (0,0) circle [radius = 0.1];
\draw [fill] (1,0) circle [radius = 0.1];
\draw [fill] (2,0) circle [radius = 0.1];
\draw [fill] (3,0) circle [radius = 0.1];
\draw [fill] (4,0) circle [radius = 0.1];
\draw [fill] (5,0) circle [radius = 0.1];
\draw [fill] (6,0) circle [radius = 0.1];
\draw [fill] (7,0) circle [radius = 0.1];
\draw [fill] (8,0) circle [radius = 0.1];
\end{tikzpicture}
\caption{The graph $G$ in the proof of Theorem~\ref{thm:two-agent-k-connected-upper} for $k=2$.}
\label{fig:ternary-tree}
\end{figure}

\begin{theorem}
\label{thm:two-agent-k-connected-upper}
Let $k$ be a positive integer.
There exists an instance with $n=2$ and identical valuations such that any allocation in which the two agents receive a total of at most $k+1$ connected pieces yields egalitarian welfare at most $\frac{1}{2}-\frac{1}{2\cdot 3^k}$.
\end{theorem} 

\begin{proof}
Let $G$ be a rooted tree with $k+2$ layers, where the first layer consists only of the root of the tree. 
The root has one child, and every vertex in subsequent layers up to the $(k+1)$st layer has three children. In particular, the $(k+2)$nd layer consists of $3^k$ leaves.
The tree $G$ for the case $k=2$ is shown in Figure~\ref{fig:ternary-tree}.
Suppose that both agents value each edge adjacent to a leaf exactly $1/3^k$ with the value distributed uniformly within the edge, and do not value any other edge.

Consider an arbitrary allocation in which the two agents receive a total of at most $k+1$ connected pieces.
We will show that the egalitarian welfare is at most $\frac{1}{2}-\frac{1}{2\cdot 3^k}$.
If there are unallocated parts of the cake, we arbitrarily allocate these parts so that the number of connected pieces that each agent receives does not increase; this does not lower the egalitarian welfare of the allocation.
Hence we may assume that the entire cake is allocated (i.e., the allocation is complete). Denote by $X_1,\dots,X_p$ the connected pieces that agent~1 receives, and $Y_1,\dots,Y_q$ the connected pieces that agent~2 receives, where $p+q\leq k+1$. 
We assume that each edge has length $1$, and refer to the distance along the (unique) path between two points in $G$ simply as the distance between these two points. Note that every connected piece has a unique point closest to the root: if there are two such points, they must be connected via a point that is strictly closer to the root than both of them. 
For every connected piece $Z$, denote by $w_z$ the unique point in $Z$ closest to the root and $u(Z)$ the utility of the piece $Z$. We will define a connected piece $Z^*$ as follows:
\begin{itemize}
    \item If $w_z$ is not a vertex of $G$, let $Z^*$ be the set of points $w$ such that the path from $w$ to the root goes through $w_z$. In other words, $Z^*$ is the part of the tree ``below'' $w_z$.
    \item If $w_z$ is a vertex of $G$, let $Z^*$ be the set of points $w$ such that the intersection of $Z$ and the path from $w$ to the root has nonzero measure. Equivalently, $Z^*$ consists of the edges adjacent to $w_z$ that have a nontrivial overlap with $Z$, along with everything ``below'' these edges.
\end{itemize}

Assume without loss of generality that agent~1 receives a piece containing the root of the tree. We claim that $u(X_1)+\dots+u(X_p) = [u(X_1^*)+\dots+u(X_p^*)]-[u(Y_1^*)+\dots+u(Y_q^*)]$. 
First, note that for each $Z^*$ where $Z\in\{X_1,\dots,X_p,Y_1,\dots,Y_q\}$, every connected piece $X_i$ and $Y_i$ is either contained in $Z^*$ in its entirety or not at all. Hence $u(Z^*)$ can be written as a sum of distinct $u(X_i)$'s and $u(Y_i)$'s. 
Moreover, one can verify from the definition that a connected piece $Z$ is contained in $W^*$ if and only if $Z=W$ or the path from $w_z$ to the root has a nontrivial overlap with $W$. 
This path alternates between pieces $X_i$ and $Y_i$ and ends with a piece $X_i$. Therefore, each $u(X_i)$ is contained in $u(X_1^*)+\dots+u(X_p^*)$ exactly once more than in $u(Y_1^*)+\dots+u(Y_q^*)$, while each $u(Y_i)$ is contained in the two sums an equal number of times. This yields the claimed equality.

Next, observe that each $u(Z^*)$ can be written as $3^s/3^k$ for some nonnegative integer $s\leq k$, or $2\cdot 3^s/3^k$ for some nonnegative integer $s\leq k-1$, or $\delta/3^k$ for some $\delta\in[0,1]$.
Note also that since agent $1$ receives a piece containing the root of the tree, for this piece $X_i$ we have $u(X_i^*) = 1$.
It follows that $[u(X_1^*)+\dots+u(X_p^*)]-[u(Y_1^*)+\dots+u(Y_q^*)]$ can be written as $1-(S+\Delta)$, where $S$ is a sum of a number of terms (say, $r$ terms, where $r\leq k$) of the form $\varepsilon\cdot 3^a$ with $\varepsilon\in\{\pm 1, \pm 2\}$ and $a$ an integer, and $|\Delta|\leq (k-r)/3^k$.
Hence, letting $d:=\left|u(X_1)+\dots+u(X_p)-\frac{1}{2}\right|$, we have
\begin{align*}
    d &= \left|\left([u(X_1^*)+\dots+u(X_p^*)]-[u(Y_1^*)+\dots+u(Y_q^*)]\right)-\frac{1}{2}\right| \\
    &= \left|(S+\Delta)-\frac{1}{2}\right| \geq \left|S-\frac{1}{2}\right|-|\Delta|  \\
    &\geq \frac{1}{2\cdot 3^r} - \frac{k-r}{3^k} = \frac{3^{k-r}-2(k-r)}{2\cdot 3^k} \geq \frac{1}{2\cdot 3^k},
\end{align*}
where the first inequality follows from the triangle inequality, the second inequality follows from Lemma~\ref{lem:powers-of-three}, and the last inequality holds since $3^b\geq 2b+1$ for any integer $b\geq 0$.
This implies that the egalitarian welfare is at most $\frac{1}{2}-\frac{1}{2\cdot 3^k}$, as claimed.
\end{proof}

Recall that the \emph{height} of a rooted tree is the length of the longest path from the root to a leaf vertex. For example, a star rooted at the center vertex has height $1$. Our next theorem shows that for graphs that can be represented as a tree of height at most $2$, we can obtain full proportionality provided that we allow two connected pieces per agent.

\begin{figure}[!ht]
\centering
\begin{tikzpicture}
\draw (1,1.5) -- (4,3) -- (7,1.5);
\draw (4,3) -- (3,1.5);
\draw (4,3) -- (5,1.5) -- (4.5,0);
\draw (0.5,0) -- (1,1.5) -- (1.5,0);
\draw (5,1.5) -- (5.5,0);
\draw (5,1.5) -- (5,0);
\draw (7,1.5) -- (7,0);
\draw [fill] (1,1.5) circle [radius = 0.1];
\draw [fill] (3,1.5) circle [radius = 0.1];
\draw [fill] (5,1.5) circle [radius = 0.1];
\draw [fill] (7,1.5) circle [radius = 0.1];
\draw [fill] (0.5,0) circle [radius = 0.1];
\draw [fill] (1.5,0) circle [radius = 0.1];
\draw [fill] (5,0) circle [radius = 0.1];
\draw [fill] (5.5,0) circle [radius = 0.1];
\draw [fill] (4.5,0) circle [radius = 0.1];
\draw [fill] (7,0) circle [radius = 0.1];
\draw [fill] (4,3) circle [radius = 0.1];
\node at (4.4,3) {$u$};
\node at (0.65,1.5) {$v_1$};
\node at (2.65,1.5) {$v_2$};
\node at (5.35,1.5) {$v_3$};
\node at (7.35,1.5) {$v_4$};
\end{tikzpicture}
\caption{Example of a graph $G$ in the proof of Theorem~\ref{thm:two-agent-two-connected-1/2}.}
\label{fig:tree-height-2}
\end{figure}
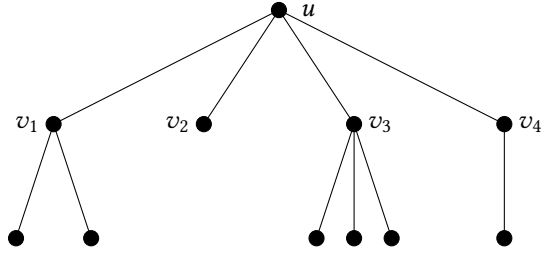

\begin{theorem}
\label{thm:two-agent-two-connected-1/2}
For $n=2$ and any graph $G$ that can be represented as a rooted tree of height at most $2$, there exists a proportional allocation such that each agent receives at most two connected pieces. 
\end{theorem}

\begin{proof}
Let $G$ be a rooted tree with root $u$, and let $v_1,\dots,v_k$ be the children of $u$ (see Figure~\ref{fig:tree-height-2}).
We move the knife in the following order: For $i=1,2,\dots,k$, we move the knife from $v_i$ down to its first child, from $v_i$ down to its second child, and so on until we reach its last child, then from $v_i$ up to $u$. 
This process ensures that the knife goes through all edges of $G$. We stop when the part already covered by the knife is worth $1/2$ to one of the agents. 
We allocate the covered part of the cake to that agent, and the remaining part to the other agent.

Clearly, the resulting allocation is proportional; it remains to show that each agent receives at most two connected pieces. We consider two cases:
\begin{itemize}
    \item Case 1: The knife stops on its way from a vertex $v_i$ to its child $w$ (possibly at $w$). This means that the first $i-1$ branches of the tree have been covered, and they are connected through $u$. 
    Moreover, the covered part in the $i$th branch are connected through $v_i$. Hence the covered part forms two connected pieces of the cake. 
    The uncovered part in the $i$th branch besides the edge $v_iw$ is connected through $v_i$, and it is connected to the remaining uncovered part from the $(i+1)$st branch onwards through $u$. 
    The uncovered part of the edge $v_iw$ forms one connected piece. It follows that both agents receive at most two connected pieces.
    \item Case 2: The knife stops on its way from a vertex $v_i$ to the root $u$ (possibly at $u$). The same argument as in Case 1 shows that the covered part forms two connected pieces of the cake. 
    The uncovered part in the $i$th branch is contained in the edge $v_iu$, and it is connected to the remaining uncovered part from the $(i+1)$st branch onwards through $u$. It follows that both agents receive at most two connected pieces.
\end{itemize}
The two cases together complete the proof.
\end{proof}

\subsection{Equitability}

We end this section by briefly considering another well-established fairness notion: equitability. Interestingly, while for approximate proportionality it is useful to consider the \emph{maximum} among the agents' values for the current piece (Theorem~\ref{thm:any-agent-egal}), for approximate equitability in the case of two agents, the appropriate quantity to consider is the \emph{sum} of these values. Note that an empty allocation is always equitable but yields the lowest possible welfare of zero, so we are interested in complete allocations.

\begin{theorem}
For $n=2$ and any graph $G$, there exists a complete and connected allocation with inequity at most $1/3$. Moreover, the bound $1/3$ is tight.
\end{theorem}

\begin{proof}
To show tightness, let $G$ be a star with three edges such that every agent values each edge exactly $1/3$, and the value is distributed uniformly within the edge. In any complete and connected allocation, one of the agents must receive two full edges (and possibly part of the third). This implies that the inequity in such an allocation is at least $1/3$.

We now prove the first part of the theorem. Let $G$ be an arbitrary graph. Our goal is to find a complete, connected allocation $(A_1,A_2)$ such that $|f_1(A_1)-f_2(A_2)|\leq 1/3$. Since the allocation is complete, we may write $f_2(A_2)=1-f_2(A_1)$. The desired condition can be rewritten as $2/3\leq f_1(A_1)+f_2(A_1)\leq 4/3$. We use a similar procedure as in Theorem~\ref{thm:any-agent-egal}, turning the graph into a tree and considering a minimal subtree that has value at least $2/3$ with respect to $f_1+f_2$. We stop either when the knife cuts a subtree $A_1$ such that $f_1(A_1)+f_2(A_1)=2/3$ (Case~1), or when a set of branches $A_1$ satisfies $f_1(A_1)+f_2(A_1)\geq 2/3$ for the first time (Case~2). An analogous argument shows that we must have $f_1(A_1)+f_2(A_1)\leq 2/3 + 2/3 = 4/3$, which yields the desired inequality.
\end{proof}

\section{Chore Division}
\label{sec:choredivision}

In this section, we assume that the graph represents a \emph{chore}, i.e., an item that yields negative value to the agents. This models, for example, a situation where we wish to divide the responsibilities of maintaining a road or cable network.

For the case of two agents, all results in cake cutting (Section~\ref{sec:twoagents}) can be translated to analogous results in chore division using a simple reduction. 
The idea is that given a chore instance, we can turn it into a cake instance by pretending that the cost functions are cake valuation functions, applying a result in the cake setting to obtain an initial allocation of the chore, and having the agents swap their assigned piece to arrive at the final allocation. 
This reduction works for translating positive results to the chore setting. For negative results, we can use the reduction in the opposite direction, starting from a chore instance and reducing it to a cake instance. 
As an illustrating example, we show how to deduce an analogue of Theorem~\ref{thm:two-agent-fixed} in the chore setting.

\begin{theorem}
\label{thm:two-agent-fixed-chore}
In chore division, for $n=2$ and any graph $G$, there exists a connected allocation such that the first agent incurs cost at most $1/2$ and the second agent incurs cost at most $2/3$.
\end{theorem}

\begin{proof}
Consider an arbitrary chore division instance. If we treat the chore valuations as cake valuations, then by Theorem~\ref{thm:two-agent-fixed}, there exists a (complete) connected allocation such that the first agent receives value at least $1/2$ and the second agent receives value at least $1/3$. 
Let the agents swap their assigned pieces in this allocation. In the resulting allocation, which is also connected, the first agent incurs cost at most $1-1/2=1/2$ and the second agent incurs cost at most $1-1/3=2/3$.
\end{proof}

When there are more than two agents, the relationship between the cake and the chore setting becomes much less clear, and we do not know how to translate results from one setting to the other. In the chore setting, we show that for each $n$, the egalitarian cost may need to be as high as $\frac{2}{n+1}$.

\begin{proposition}
\label{prop:any-agent-chore-negative}
In chore division, there exists a graph $G$ and identical valuations of the agents such that any connected allocation yields egalitarian cost at least $\frac{2}{n+1}$.
\end{proposition}

\begin{proof}
Let $G$ be a star with $n+1$ edges such that every agent has cost exactly $\frac{1}{n+1}$ for each edge, and the cost is distributed uniformly within the edge. 
Assume for contradiction that there is a connected allocation with egalitarian cost strictly less than $\frac{2}{n+1}$. Let $v_1,\dots,v_{n+1}$ be the endpoints of the edges different from the center of the star. 
For each $i$, some agent must be allocated a piece containing $v_i$. If an agent receives a piece containing two endpoints, she incurs cost at least $\frac{2}{n+1}$, which is impossible. 
So every agent's piece contains at most one endpoint $v_i$. However, since there are $n$ agents and $n+1$ endpoints, some endpoint is left unallocated, a contradiction.
\end{proof}

The bound $\frac{2}{n+1}$ is tight for $n=2$ due to Theorem~\ref{thm:two-agent-fixed-chore}. 
Next, we show that it remains tight as long as $n\leq 5$. 
A simpler protocol for the case $n=3$ is given in Appendix~\ref{app:three-agent-chore}.

\begin{theorem}
In chore division, for $n\leq 5$ and any graph $G$, there exists a connected allocation with egalitarian cost at most $\frac{2}{n+1}$.
\end{theorem}

\begin{proof}
We proceed by strong induction on $n$; the case $n=1$ is trivial while the case $n=2$ follows from Theorem~\ref{thm:two-agent-fixed-chore}. 
Let $3\leq n\leq 5$. For an arbitrary piece of chore, let $c_1\leq c_2\leq\dots\leq c_n$ be the costs of the agents for the piece in increasing order. We define two conditions that the piece may satisfy:
\begin{itemize}
    \item Condition 1: $c_1\leq \frac{1}{n+1}$, and $c_i\leq\frac{i-1}{n+1}$ for $i=2,3,\dots,n$.
    \item Condition 2: $c_i>\frac{i+1}{n+1}$ for $i=1,2,\dots,n-1$, and $c_n>\frac{n}{n+1}$.
\end{itemize}

We turn the graph $G$ into a tree as in Theorem~\ref{thm:any-agent-egal}. Choose an arbitrary vertex $u$ of the tree $G$ as its root. 
Let $v$ be a vertex such that the subtree rooted at $v$ does not satisfy Condition~1, but the subtree rooted at any child of $v$ does. Let $w_1,\dots,w_k$ be the children of $v$. We consider two cases:

\begin{itemize}
    \item Case 1: At least one of the $k$ branches of $v$ along with the corresponding subtree does not satisfy Condition~1. Assume without loss of generality that the branch containing $w_1$ is one such branch. 
    By our assumption, the subtree rooted at $w_1$ satisfies Condition~1. Hence, by moving a knife from $w_1$ to $v$, we can find the point $x$ such that if we consider the interval $[w_1,x]$ together with the subtree rooted at $w_1$ as a subtree rooted at $x$, at least one of the inequalities in Condition~1 is tight for this subtree and the remaining inequalities still hold. Assume that for each $j=1,2,\dots,n$, agent $j$ has cost $c_j$ for this subtree.
    
    Suppose that the inequality involving $c_i$ is tight. If $i=1$, we allocate this subtree to agent~1, who incurs cost $\frac{1}{n+1}$. 
    The rest of the chore has cost at most $\frac{n}{n+1}$ to the remaining agents, and by the inductive hypothesis, it can be allocated in a connected manner to these agents so that each agent incurs cost at most $\frac{2}{n}\cdot\frac{n}{n+1}=\frac{2}{n+1}$. Suppose now that $i\geq 2$. 
    This means that $c_i=\frac{i-1}{n+1}$. The subtree has cost at most $\frac{i-1}{n+1}$ to the first $i-1$ agents. By the inductive hypothesis, it can be allocated in a connected manner to these agents so that each agent incurs cost at most $\frac{2}{i}\cdot\frac{i-1}{n+1}<\frac{2}{n+1}$. 
    The rest of the chore has cost at most $\frac{n-i+2}{n+1}$ to the remaining $n-i+1$ agents. By the inductive hypothesis, it can be allocated in a connected manner to these agents so that each agent incurs cost at most $\frac{2}{n-i+2}\cdot\frac{n-i+2}{n+1}=\frac{2}{n+1}$.
    Hence we have a connected allocation with egalitarian cost at most $\frac{2}{n+1}$.
    
    \item Case 2: Every branch of $v$ along with the corresponding subtree satisfies Condition~1. Let $t\in\{1,2,\dots,k\}$ be the smallest number such that the first $t$ branches and their subtrees together, which we denote by $T$, do not satisfy Condition~1. 
    In particular, the first $t-1$ branches and their subtrees together, which we denote by $T_1$, satisfy Condition~1, and the $t$th branch and its subtree together, which we denote by $T_2$, also satisfy this condition. 
    
    We claim that $T$ does not satisfy Condition~2. Assume for contradiction that the opposite is true. Let $a_1\leq\dots\leq a_n$ be the costs of the agents for $T_1$ in increasing order, and $b_1\leq\dots\leq b_n$ be the corresponding costs for $T_2$. By Condition~1, we have $a_1\leq\frac{1}{n+1}$ and $a_i\leq\frac{i-1}{n+1}$ for $i\geq 2$; analogous upper bounds hold for the $b_i$'s. In order for a sum $a_i+b_j$ to be strictly greater than $\frac{r}{n+1}$ for some positive integer $r$, the upper bounds of $a_i$ and $b_j$ must add up to at least $\frac{r+1}{n+1}$. Hence, in order for Condition~2 to be satisfied, we must have
    $$2\left(\frac{1}{n+1}+\sum_{i=1}^{n-1}\frac{i}{n+1}\right)\geq \sum_{i=3}^{n+1}\frac{i}{n+1} + \frac{n+1}{n+1}.$$
    Multiplying both sides by $n+1$, this is equivalent to
    $$2+(n-1)n \geq \frac{(n+1)(n+2)}{2} -3 + (n+1),$$
    or $(n-1)(n-6)\geq 0$, which is false for $4\leq n\leq 5$. So $T$ does not satisfy Condition~2.
    
    Recall that $T$ does not satisfy Condition~1. Let $c_1\leq\dots\leq c_n$ be the costs of the agents for $T$ in increasing order, and let $i$ be the smallest index for which the inequality involving $c_i$ in Condition~1 fails. If $i\geq 2$, we may proceed as in Case~1 by allocating $T$ to the first $i-1$ agents and the rest of the chore to the remaining $n-i+1$ agents. So we may assume that $i=1$, i.e., $c_1\geq\frac{1}{n+1}$. Next, let $j$ be the smallest index for which the inequality involving $c_j$ in Condition~2 fails. Since $c_n\geq c_{n-1}$, we have $j< n$. This means that $c_r>\frac{r+1}{n+1}$ for $r=1,2,\dots,j-1$ and $c_j \leq \frac{j+1}{n+1}$. Hence, $T$ has cost at most $\frac{j+1}{n+1}$ to the first $j$ agents. By the inductive hypothesis, it can be allocated in a connected manner to these agents so that each agent incurs cost at most $\frac{2}{j+1}\cdot\frac{j+1}{n+1}=\frac{2}{n+1}$. If $j\geq 2$, then since $c_{j-1}>\frac{j}{n+1}$, the rest of the chore has cost at most $\frac{n-j+1}{n+1}$ to the remaining $n-j$ agents. If $j=1$, we know that $c_1\geq\frac{1}{n+1}$, and so the rest of the chore has cost at most $\frac{n}{n+1}=\frac{n-j+1}{n+1}$ to the remaining $n-1$ agents. In either case, by the inductive hypothesis, the rest of the chore can be allocated in a connected manner to these agents so that each agent incurs cost at most $\frac{2}{n-j+1}\cdot\frac{n-j+1}{n+1}=\frac{2}{n+1}$. Hence we again have a connected allocation with egalitarian cost at most $\frac{2}{n+1}$.
\end{itemize}
The two cases together complete the proof.
\end{proof}

We conjecture that the bound $\frac{2}{n+1}$ is tight for all $n$, and leave it as an intriguing open question.

\section{Conclusion and Future Work}

In this paper, we introduce and study a generalized version of the classical cake-cutting problem, where the cake can be represented by an arbitrary graph instead of an interval. 
We establish bounds on the utilities that can be guaranteed to the agents for various classes of graphs, both for cake cutting and chore division, and demonstrate in several cases that our guarantees are tight. 
We also show that better guarantees are possible if we allow more connected pieces per agent, and exhibit an algorithm that computes an approximately equitable allocation.

Our work opens up a number of new directions for future research. Besides proportionality and equitability, another prominent fairness notion is \emph{envy-freeness}, which stipulates that no agent prefers another agent's bundle to her own in the allocation. 
In the case of two agents, envy-freeness and proportionality are equivalent, and approximate proportionality bounds readily translate to corresponding approximate envy-freeness results.
However, this equivalence ceases to hold when there are more than two agents.
If the graph consists of a single edge, a connected envy-free allocation always exists for any number of agents \citep{Stromquist80}.
It would be interesting to see whether one can obtain (approximate) envy-freeness guarantees for different classes of graphs.

Like in the vast majority of the fair division literature, we assume in this paper that all parts of the resource either yield nonnegative utility to every agent (cake cutting) or nonpositive utility to every agent (chore division). 
Recently, Bogomolnaia et al.~\cite{BogomolnaiaMoSa17} and Segal-Halevi~\cite{Segalhalevi18} considered a generalization where an agent may have positive utility for some parts of the resource and negative utility for other parts, and different agents may have different evaluations. Aziz et al.~\cite{AzizCaIg19} showed the existence of a connected proportional allocation in this general setting when the resource is represented by an interval. Again, extending this result to more complex graphs is an appealing direction that we leave for future work.


\bibliographystyle{ACM-Reference-Format}
\bibliography{main}


\begin{thebibliography}{25}


\ifx \showCODEN    \undefined \def \showCODEN     #1{\unskip}     \fi
\ifx \showDOI      \undefined \def \showDOI       #1{#1}\fi
\ifx \showISBNx    \undefined \def \showISBNx     #1{\unskip}     \fi
\ifx \showISBNxiii \undefined \def \showISBNxiii  #1{\unskip}     \fi
\ifx \showISSN     \undefined \def \showISSN      #1{\unskip}     \fi
\ifx \showLCCN     \undefined \def \showLCCN      #1{\unskip}     \fi
\ifx \shownote     \undefined \def \shownote      #1{#1}          \fi
\ifx \showarticletitle \undefined \def \showarticletitle #1{#1}   \fi
\ifx \showURL      \undefined \def \showURL       {\relax}        \fi
\providecommand\bibfield[2]{#2}
\providecommand\bibinfo[2]{#2}
\providecommand\natexlab[1]{#1}
\providecommand\showeprint[2][]{arXiv:#2}

\bibitem[\protect\citeauthoryear{Aumann and Dombb}{Aumann and Dombb}{2015}]%
        {AumannDo15}
\bibfield{author}{\bibinfo{person}{Yonatan Aumann} {and} \bibinfo{person}{Yair
  Dombb}.} \bibinfo{year}{2015}\natexlab{}.
\newblock \showarticletitle{The efficiency of fair division with connected
  pieces}.
\newblock \bibinfo{journal}{\emph{ACM Transactions on Economics and
  Computation}} \bibinfo{volume}{3}, \bibinfo{number}{4}
  (\bibinfo{year}{2015}), \bibinfo{pages}{23:1--23:16}.
\newblock


\bibitem[\protect\citeauthoryear{Aumann, Dombb, and Hassidim}{Aumann
  et~al\mbox{.}}{2013}]%
        {AumannDoHa13}
\bibfield{author}{\bibinfo{person}{Yonatan Aumann}, \bibinfo{person}{Yair
  Dombb}, {and} \bibinfo{person}{Avinatan Hassidim}.}
  \bibinfo{year}{2013}\natexlab{}.
\newblock \showarticletitle{Computing socially-efficient cake divisions}. In
  \bibinfo{booktitle}{\emph{Proceedings of the 12th International Conference on
  Autonomous Agents and Multiagent Systems (AAMAS)}}.
  \bibinfo{pages}{343--350}.
\newblock


\bibitem[\protect\citeauthoryear{Aziz, Caragiannis, Igarashi, and Walsh}{Aziz
  et~al\mbox{.}}{2019}]%
        {AzizCaIg19}
\bibfield{author}{\bibinfo{person}{Haris Aziz}, \bibinfo{person}{Ioannis
  Caragiannis}, \bibinfo{person}{Ayumi Igarashi}, {and} \bibinfo{person}{Toby
  Walsh}.} \bibinfo{year}{2019}\natexlab{}.
\newblock \showarticletitle{Fair allocation of indivisible goods and chores}.
  In \bibinfo{booktitle}{\emph{Proceedings of the 28th International Joint
  Conference on Artificial Intelligence (IJCAI)}}. \bibinfo{pages}{53--59}.
\newblock


\bibitem[\protect\citeauthoryear{Bei, Chen, Hua, Tao, and Yang}{Bei
  et~al\mbox{.}}{2012}]%
        {BeiChHu12}
\bibfield{author}{\bibinfo{person}{Xiaohui Bei}, \bibinfo{person}{Ning Chen},
  \bibinfo{person}{Xia Hua}, \bibinfo{person}{Biaoshuai Tao}, {and}
  \bibinfo{person}{Endong Yang}.} \bibinfo{year}{2012}\natexlab{}.
\newblock \showarticletitle{Optimal proportional cake cutting with connected
  pieces}. In \bibinfo{booktitle}{\emph{Proceedings of the 26th AAAI Conference
  on Artificial Intelligence (AAAI)}}. \bibinfo{pages}{1263--1269}.
\newblock


\bibitem[\protect\citeauthoryear{Bei, Igarashi, Lu, and Suksompong}{Bei
  et~al\mbox{.}}{2019}]%
        {BeiIgLu19}
\bibfield{author}{\bibinfo{person}{Xiaohui Bei}, \bibinfo{person}{Ayumi
  Igarashi}, \bibinfo{person}{Xinhang Lu}, {and} \bibinfo{person}{Warut
  Suksompong}.} \bibinfo{year}{2019}\natexlab{}.
\newblock \showarticletitle{Connected fair allocation of indivisible goods}.
\newblock \bibinfo{journal}{\emph{CoRR}}  \bibinfo{volume}{abs/1908.05433}
  (\bibinfo{year}{2019}).
\newblock


\bibitem[\protect\citeauthoryear{Bil\`{o}, Caragiannis, Flammini, Igarashi,
  Monaco, Peters, Vinci, and Zwicker}{Bil\`{o} et~al\mbox{.}}{2019}]%
        {BiloCaFl19}
\bibfield{author}{\bibinfo{person}{Vittorio Bil\`{o}}, \bibinfo{person}{Ioannis
  Caragiannis}, \bibinfo{person}{Michele Flammini}, \bibinfo{person}{Ayumi
  Igarashi}, \bibinfo{person}{Gianpiero Monaco}, \bibinfo{person}{Dominik
  Peters}, \bibinfo{person}{Cosimo Vinci}, {and} \bibinfo{person}{William~S.
  Zwicker}.} \bibinfo{year}{2019}\natexlab{}.
\newblock \showarticletitle{Almost envy-free allocations with connected
  bundles}. In \bibinfo{booktitle}{\emph{Proceedings of the 10th Innovations in
  Theoretical Computer Science Conference (ITCS)}}.
  \bibinfo{pages}{14:1--14:21}.
\newblock


\bibitem[\protect\citeauthoryear{Bogomolnaia, Moulin, Sandomirskiy, and
  Yanovskaya}{Bogomolnaia et~al\mbox{.}}{2017}]%
        {BogomolnaiaMoSa17}
\bibfield{author}{\bibinfo{person}{Anna Bogomolnaia},
  \bibinfo{person}{Herv\'{e} Moulin}, \bibinfo{person}{Fedor Sandomirskiy},
  {and} \bibinfo{person}{Elena Yanovskaya}.} \bibinfo{year}{2017}\natexlab{}.
\newblock \showarticletitle{Competitive division of a mixed manna}.
\newblock \bibinfo{journal}{\emph{Econometrica}} \bibinfo{volume}{85},
  \bibinfo{number}{6} (\bibinfo{year}{2017}), \bibinfo{pages}{1847--1871}.
\newblock


\bibitem[\protect\citeauthoryear{Bouveret, Cechl\'{a}rov\'{a}, Elkind,
  Igarashi, and Peters}{Bouveret et~al\mbox{.}}{2017}]%
        {BouveretCeEl17}
\bibfield{author}{\bibinfo{person}{Sylvain Bouveret},
  \bibinfo{person}{Katar\'{i}na Cechl\'{a}rov\'{a}}, \bibinfo{person}{Edith
  Elkind}, \bibinfo{person}{Ayumi Igarashi}, {and} \bibinfo{person}{Dominik
  Peters}.} \bibinfo{year}{2017}\natexlab{}.
\newblock \showarticletitle{Fair division of a graph}. In
  \bibinfo{booktitle}{\emph{Proceedings of the 26th International Joint
  Conference on Artificial Intelligence (IJCAI)}}. \bibinfo{pages}{135--141}.
\newblock


\bibitem[\protect\citeauthoryear{Bouveret, Cechl\'{a}rov\'{a}, and
  Lesca}{Bouveret et~al\mbox{.}}{2019}]%
        {BouveretCeLe19}
\bibfield{author}{\bibinfo{person}{Sylvain Bouveret},
  \bibinfo{person}{Katar\'{i}na Cechl\'{a}rov\'{a}}, {and}
  \bibinfo{person}{Julien Lesca}.} \bibinfo{year}{2019}\natexlab{}.
\newblock \showarticletitle{Chore division on a graph}.
\newblock \bibinfo{journal}{\emph{Autonomous Agents and Multi-Agent Systems}}
  \bibinfo{volume}{33}, \bibinfo{number}{5} (\bibinfo{year}{2019}),
  \bibinfo{pages}{540--563}.
\newblock


\bibitem[\protect\citeauthoryear{Brams and Taylor}{Brams and Taylor}{1996}]%
        {BramsTa96}
\bibfield{author}{\bibinfo{person}{Steven~J. Brams} {and}
  \bibinfo{person}{Alan~D. Taylor}.} \bibinfo{year}{1996}\natexlab{}.
\newblock \bibinfo{booktitle}{\emph{Fair Division: From Cake-Cutting to Dispute
  Resolution}}.
\newblock \bibinfo{publisher}{Cambridge University Press}.
\newblock


\bibitem[\protect\citeauthoryear{Cechl\'{a}rov\'{a}, Dobo\v{s}, and
  Pill\'{a}rov\'{a}}{Cechl\'{a}rov\'{a} et~al\mbox{.}}{2013}]%
        {CechlarovaDoPi13}
\bibfield{author}{\bibinfo{person}{Katar\'{i}na Cechl\'{a}rov\'{a}},
  \bibinfo{person}{Jozef Dobo\v{s}}, {and} \bibinfo{person}{Eva
  Pill\'{a}rov\'{a}}.} \bibinfo{year}{2013}\natexlab{}.
\newblock \showarticletitle{On the existence of equitable cake divisions}.
\newblock \bibinfo{journal}{\emph{Information Sciences}}  \bibinfo{volume}{228}
  (\bibinfo{year}{2013}), \bibinfo{pages}{239--245}.
\newblock


\bibitem[\protect\citeauthoryear{Cechl\'{a}rov\'{a} and
  Pill\'{a}rov\'{a}}{Cechl\'{a}rov\'{a} and Pill\'{a}rov\'{a}}{2012}]%
        {CechlarovaPi12}
\bibfield{author}{\bibinfo{person}{Katar\'{i}na Cechl\'{a}rov\'{a}} {and}
  \bibinfo{person}{Eva Pill\'{a}rov\'{a}}.} \bibinfo{year}{2012}\natexlab{}.
\newblock \showarticletitle{On the computability of equitable divisions}.
\newblock \bibinfo{journal}{\emph{Discrete Optimization}} \bibinfo{volume}{9},
  \bibinfo{number}{4} (\bibinfo{year}{2012}), \bibinfo{pages}{249--257}.
\newblock


\bibitem[\protect\citeauthoryear{Dubins and Spanier}{Dubins and
  Spanier}{1961}]%
        {DubinsSp61}
\bibfield{author}{\bibinfo{person}{Lester~E. Dubins} {and}
  \bibinfo{person}{Edwin~H. Spanier}.} \bibinfo{year}{1961}\natexlab{}.
\newblock \showarticletitle{How to cut a cake fairly}.
\newblock \bibinfo{journal}{\emph{The American Mathematical Monthly}}
  \bibinfo{volume}{68}, \bibinfo{number}{1} (\bibinfo{year}{1961}),
  \bibinfo{pages}{1--17}.
\newblock


\bibitem[\protect\citeauthoryear{Igarashi and Peters}{Igarashi and
  Peters}{2019}]%
        {IgarashiPe19}
\bibfield{author}{\bibinfo{person}{Ayumi Igarashi} {and}
  \bibinfo{person}{Dominik Peters}.} \bibinfo{year}{2019}\natexlab{}.
\newblock \showarticletitle{Pareto-optimal allocation of indivisible goods with
  connectivity constraints}. In \bibinfo{booktitle}{\emph{Proceedings of the
  33rd AAAI Conference on Artificial Intelligence (AAAI)}}.
  \bibinfo{pages}{2045--2052}.
\newblock


\bibitem[\protect\citeauthoryear{Lonc and Truszczynski}{Lonc and
  Truszczynski}{2018}]%
        {LoncTr18}
\bibfield{author}{\bibinfo{person}{Zbigniew Lonc} {and}
  \bibinfo{person}{Miroslaw Truszczynski}.} \bibinfo{year}{2018}\natexlab{}.
\newblock \showarticletitle{Maximin share allocations on cycles}. In
  \bibinfo{booktitle}{\emph{Proceedings of the 27th International Joint
  Conference on Artificial Intelligence (IJCAI)}}. \bibinfo{pages}{410--416}.
\newblock


\bibitem[\protect\citeauthoryear{Procaccia}{Procaccia}{2016}]%
        {Procaccia16}
\bibfield{author}{\bibinfo{person}{Ariel~D. Procaccia}.}
  \bibinfo{year}{2016}\natexlab{}.
\newblock \showarticletitle{Cake Cutting Algorithms}.
\newblock In \bibinfo{booktitle}{\emph{Handbook of Computational Social
  Choice}}, \bibfield{editor}{\bibinfo{person}{Felix Brandt},
  \bibinfo{person}{Vincent Conitzer}, \bibinfo{person}{Ulle Endriss},
  \bibinfo{person}{J\'{e}r\^{o}me Lang}, {and} \bibinfo{person}{Ariel~D.
  Procaccia}} (Eds.). \bibinfo{publisher}{Cambridge University Press},
  Chapter~13, \bibinfo{pages}{311--329}.
\newblock


\bibitem[\protect\citeauthoryear{Robertson and Webb}{Robertson and
  Webb}{1998}]%
        {RobertsonWe98}
\bibfield{author}{\bibinfo{person}{Jack Robertson} {and}
  \bibinfo{person}{William Webb}.} \bibinfo{year}{1998}\natexlab{}.
\newblock \bibinfo{booktitle}{\emph{Cake-Cutting Algorithms: Be Fair if You
  Can}}.
\newblock \bibinfo{publisher}{Peters/CRC Press}.
\newblock


\bibitem[\protect\citeauthoryear{Segal-Halevi}{Segal-Halevi}{2018}]%
        {Segalhalevi18}
\bibfield{author}{\bibinfo{person}{Erel Segal-Halevi}.}
  \bibinfo{year}{2018}\natexlab{}.
\newblock \showarticletitle{Fairly dividing a cake after some parts were burnt
  in the oven}. In \bibinfo{booktitle}{\emph{Proceedings of the 17th
  International Conference on Autonomous Agents and Multiagent Systems
  (AAMAS)}}. \bibinfo{pages}{1276--1284}.
\newblock


\bibitem[\protect\citeauthoryear{Segal-Halevi, Hassidim, and
  Aumann}{Segal-Halevi et~al\mbox{.}}{2016}]%
        {SegalhaleviHaAu16}
\bibfield{author}{\bibinfo{person}{Erel Segal-Halevi},
  \bibinfo{person}{Avinatan Hassidim}, {and} \bibinfo{person}{Yonatan Aumann}.}
  \bibinfo{year}{2016}\natexlab{}.
\newblock \showarticletitle{Waste makes haste: bounded time algorithms for
  envy-free cake cutting with free disposal}.
\newblock \bibinfo{journal}{\emph{ACM Transactions on Algorithms}}
  \bibinfo{volume}{13}, \bibinfo{number}{1} (\bibinfo{year}{2016}),
  \bibinfo{pages}{12:1--12:32}.
\newblock


\bibitem[\protect\citeauthoryear{Segal-Halevi, Nitzan, Hassidim, and
  Aumann}{Segal-Halevi et~al\mbox{.}}{2017}]%
        {SegalhaleviNiHa17}
\bibfield{author}{\bibinfo{person}{Erel Segal-Halevi}, \bibinfo{person}{Shmuel
  Nitzan}, \bibinfo{person}{Avinatan Hassidim}, {and} \bibinfo{person}{Yonatan
  Aumann}.} \bibinfo{year}{2017}\natexlab{}.
\newblock \showarticletitle{Fair and square: cake-cutting in two dimensions}.
\newblock \bibinfo{journal}{\emph{Journal of Mathematical Economics}}
  \bibinfo{volume}{70}, \bibinfo{number}{8} (\bibinfo{year}{2017}),
  \bibinfo{pages}{1--28}.
\newblock


\bibitem[\protect\citeauthoryear{Steinhaus}{Steinhaus}{1948}]%
        {Steinhaus48}
\bibfield{author}{\bibinfo{person}{Hugo Steinhaus}.}
  \bibinfo{year}{1948}\natexlab{}.
\newblock \showarticletitle{The problem of fair division}.
\newblock \bibinfo{journal}{\emph{Econometrica}} \bibinfo{volume}{16},
  \bibinfo{number}{1} (\bibinfo{year}{1948}), \bibinfo{pages}{101--104}.
\newblock


\bibitem[\protect\citeauthoryear{Stromquist}{Stromquist}{1980}]%
        {Stromquist80}
\bibfield{author}{\bibinfo{person}{Walter Stromquist}.}
  \bibinfo{year}{1980}\natexlab{}.
\newblock \showarticletitle{How to cut a cake fairly}.
\newblock \bibinfo{journal}{\emph{The American Mathematical Monthly}}
  \bibinfo{volume}{87}, \bibinfo{number}{8} (\bibinfo{year}{1980}),
  \bibinfo{pages}{640--644}.
\newblock


\bibitem[\protect\citeauthoryear{Stromquist}{Stromquist}{2008}]%
        {Stromquist08}
\bibfield{author}{\bibinfo{person}{Walter Stromquist}.}
  \bibinfo{year}{2008}\natexlab{}.
\newblock \showarticletitle{Envy-free cake divisions cannot be found by finite
  protocols}.
\newblock \bibinfo{journal}{\emph{The Electronic Journal of Combinatorics}}
  \bibinfo{volume}{15} (\bibinfo{year}{2008}), \bibinfo{pages}{\#R11}.
\newblock


\bibitem[\protect\citeauthoryear{Su}{Su}{1999}]%
        {Su99}
\bibfield{author}{\bibinfo{person}{Francis~Edward Su}.}
  \bibinfo{year}{1999}\natexlab{}.
\newblock \showarticletitle{Rental harmony: Sperner's lemma in fair division}.
\newblock \bibinfo{journal}{\emph{The American Mathematical Monthly}}
  \bibinfo{volume}{106}, \bibinfo{number}{10} (\bibinfo{year}{1999}),
  \bibinfo{pages}{930--942}.
\newblock


\bibitem[\protect\citeauthoryear{Suksompong}{Suksompong}{2019}]%
        {Suksompong19}
\bibfield{author}{\bibinfo{person}{Warut Suksompong}.}
  \bibinfo{year}{2019}\natexlab{}.
\newblock \showarticletitle{Fairly allocating contiguous blocks of indivisible
  items}.
\newblock \bibinfo{journal}{\emph{Discrete Applied Mathematics}}
  \bibinfo{volume}{260} (\bibinfo{year}{2019}), \bibinfo{pages}{227--236}.
\newblock


\end{thebibliography}

\appendix

\section{Chore Division Protocol for Three Agents}
\label{app:three-agent-chore}

\begin{proposition}
\label{prop:three-agent-chore}
In chore division, for $n=3$ and any graph $G$, there exists a connected allocation with egalitarian cost at most $1/2$.
\end{proposition}


\begin{proof}
Pick two arbitrary agents. By Theorem~\ref{thm:two-agent-fixed-chore}, there exists a connected allocation to the two agents such that the first agent incurs cost at most $1/2$ and the second agent incurs cost at most $2/3$. 
Fix the piece assigned to the first agent, and divide the piece assigned to the second agent further between the second and third agents. By Theorem~\ref{thm:two-agent-fixed-chore} again, there exists a connected allocation of the latter piece such that the third agent incurs cost at most $1/2$ and the second agent incurs cost at most $2/3\times 2/3 = 4/9 < 1/2$. 
Hence the egalitarian cost of the resulting allocation is at most $1/2$.
\end{proof}

\end{document}